\documentclass[11pt]{article}

\usepackage[a4paper,margin=1in]{geometry}
\usepackage{amsmath,amssymb,amsthm,mathtools}
\usepackage{mathrsfs}
\usepackage{bm}
\usepackage{needspace}
\usepackage{microtype}
\usepackage{tikz}
\usetikzlibrary{
  arrows.meta,
  decorations.pathreplacing,
  positioning
}
\usepackage[
  backend=biber,
  style=numeric,
  sorting=nyt,
  giveninits=true,
  maxbibnames=99,
  doi=true,
  eprint=true,
  url=true,
  isbn=false
]{biblatex}
\AtBeginBibliography{\small}
\usepackage[hidelinks]{hyperref}
\hypersetup{pdftitle={The Failure of Simultaneous Quantum Typicality},pdfauthor={Benoit Collins}}
\usepackage[nameinlink,capitalize,noabbrev]{cleveref}

\allowdisplaybreaks
\numberwithin{equation}{section}

\newtheorem{theorem}{Theorem}[section]
\newtheorem{proposition}[theorem]{Proposition}
\newtheorem{lemma}[theorem]{Lemma}
\newtheorem{corollary}[theorem]{Corollary}
\newtheorem{conjecture}[theorem]{Conjecture}
\newtheorem{remark}[theorem]{Remark}
\newtheorem{definition}[theorem]{Definition}
\newtheorem*{informalconjecture}{Conjecture}
\newtheorem*{informalresult}{Theorem}
\crefname{conjecture}{Conjecture}{Conjectures}
\Crefname{conjecture}{Conjecture}{Conjectures}

\newcommand{\cH}{\mathcal H}
\newcommand{\cD}{\mathcal D}
\newcommand{\cS}{\mathcal S}
\newcommand{\cK}{\mathcal K}
\newcommand{\cE}{\mathcal E}
\newcommand{\cN}{\mathcal N}
\newcommand{\1}{\mathbf 1}
\newcommand{\Tr}{\operatorname{Tr}}
\newcommand{\rank}{\operatorname{rank}}
\newcommand{\supp}{\operatorname{supp}}
\newcommand{\Sym}{\operatorname{Sym}}
\newcommand{\Inv}{\operatorname{Inv}}
\newcommand{\CP}{\mathbb {CP}}
\newcommand{\eps}{\varepsilon}
\newcommand{\ket}[1]{|#1\rangle}
\newcommand{\ketbra}[2]{|#1\rangle\!\langle#2|}
\newcommand{\proj}[1]{|#1\rangle\!\langle#1|}
\newcommand{\norm}[1]{\left\lVert #1\right\rVert}
\newcommand{\abs}[1]{\left|#1\right|}
\newcommand{\e}{\mathrm e}
\newcommand{\dd}{\,\mathrm d}

\definecolor{qtblue}{HTML}{3B6FB6}
\definecolor{qtorange}{HTML}{D5792A}
\definecolor{qtgreen}{HTML}{34856F}
\definecolor{qtred}{HTML}{B24A4A}
\definecolor{qtgray}{HTML}{59636E}
\tikzset{
  qtbox/.style={
    draw=qtblue!80!black,
    fill=qtblue!7,
    rounded corners=2pt,
    thick,
    align=center,
    inner sep=5pt
  },
  qtgreenbox/.style={
    draw=qtgreen!85!black,
    fill=qtgreen!8,
    rounded corners=2pt,
    thick,
    align=center,
    inner sep=5pt
  },
  qtorangebox/.style={
    draw=qtorange!90!black,
    fill=qtorange!10,
    rounded corners=2pt,
    thick,
    align=center,
    inner sep=5pt
  },
  qtredbox/.style={
    draw=qtred!90!black,
    fill=qtred!8,
    rounded corners=2pt,
    thick,
    align=center,
    inner sep=5pt
  },
  qtarrow/.style={
    -{Latex[length=2.2mm,width=1.6mm]},
    thick,
    draw=qtgray
  },
  qtsoftarrow/.style={
    -{Latex[length=2mm,width=1.4mm]},
    semithick,
    draw=qtgray!75
  },
  qtnote/.style={
    align=center,
    font=\scriptsize,
    text=qtgray
  }
}

\title{\bfseries The Failure of Simultaneous Quantum Typicality}
\author{Beno{\^\i}t Collins\\[0.4em]
  \normalsize Department of Mathematics, Graduate School of Science\\
  \normalsize Kyoto University, Kyoto 606-8502, Japan\\
  \normalsize \href{mailto:collins@math.kyoto-u.ac.jp}%
  {\texttt{collins@math.kyoto-u.ac.jp}}}
\date{\today}

\begin{document}
\maketitle

\begin{abstract}
We establish exact thresholds for the failure of simultaneous quantum
typicality: three parties for mixed states and four for pure states.  For
the four-qubit Higuchi--Sudbery state, every sequence satisfying the
expected purity bounds on three overlapping marginals, with sufficiently
small fixed positive entropy slack, is asymptotically perfectly
distinguishable from the tensor powers of the input: its trace distance
from them tends to the maximum value $2$.  This maximal
separation persists, with a common positive slack, throughout an open
neighborhood in the full four-qubit state space.  Tracing out one party
yields a mixed tripartite example.

The proof connects simultaneous quantum smoothing with polynomial
concentration on $\CP^1$.  Symmetry reduces the problem to a frame
inequality on $\Sym^n(\mathbb C^2)$, where a linear entropy deficit
competes with a quadratic cost of polynomial concentration.
Bernstein--Markov estimates and a subharmonic barrier control this cost,
yielding an exponential obstruction to simultaneous typicality.

These results disprove Dutil's multiparty quantum typicality conjecture,
formulated in 2011, and the one-shot simultaneous min-entropy-smoothing
conjectures of Drescher--Fawzi and Colomer--Winter.  Together with the
positive results for smaller systems, they determine the smallest numbers
of parties at which simultaneous typicality can fail.
\end{abstract}

\section{Introduction}

Simultaneous quantum typicality asks whether a single small perturbation
of many copies of a state can give every marginal the entropy rate
prescribed by the input.  We determine the smallest numbers of parties
for which this property can fail: three for mixed states and four for
pure states.  The failure takes the strongest possible form: for our examples, states
satisfying the required marginal purity bounds, with sufficiently small
fixed positive slack, are asymptotically perfectly distinguishable from
the tensor-power input: their trace distance from it tends to $2$.  For
four-qubit inputs this maximal separation holds
throughout an open neighborhood in the full state space.

Let $p:\mathcal X\to[0,1]$ be a probability distribution on a finite
alphabet $\mathcal X$, with \emph{Shannon entropy}
$H(p)=-\sum_{x\in\mathcal X}p(x)\log p(x)$, where $0\log0=0$.
All logarithms are natural.
Let $p^{\otimes n}$ be the law on $\mathcal X^n$ of $n$ independent
samples.  For
$\delta>0$, define the \emph{weakly typical set}
\[
 \mathcal T_{p,\delta}^{(n)}
 :=\left\{x^n\in\mathcal X^n:\left|-\frac1n\log p^{\otimes n}(x^n)-H(p)\right|
 \leq\delta\right\}.
\]
The \emph{asymptotic equipartition property} states that
$p^{\otimes n}(\mathcal T_{p,\delta}^{(n)})\to1$.  Moreover, every
$x^n\in\mathcal T_{p,\delta}^{(n)}$ satisfies the pointwise bounds
\[
 \e^{-n(H(p)+\delta)}
 \leq p^{\otimes n}(x^n)
 \leq\e^{-n(H(p)-\delta)},
\]
and consequently
\[
 p^{\otimes n}(\mathcal T_{p,\delta}^{(n)})\,\e^{n(H(p)-\delta)}
 \leq\abs{\mathcal T_{p,\delta}^{(n)}}
 \leq\e^{n(H(p)+\delta)}.
\]
Thus ``approximately'' means equality at exponential scale under the product
law $p^{\otimes n}$: the exponent of each typical-word probability differs
from $H(p)$ by at most $\delta$, while
$n^{-1}\log \abs{\mathcal T_{p,\delta}^{(n)}}$ differs from $H(p)$ by at
most $\delta+o_n(1)$.  See \cite[Chap.~3]{CoverThomas} for the
asymptotic equipartition property, which goes back to \cite{Shannon}, and
\cite[Secs.~2.4--2.5]{ElGamalKim} for typical and jointly typical sequences.

The \emph{quantum asymptotic equipartition property} replaces a typical set by a
\emph{typical subspace}.  For a density operator
$\rho\in\operatorname{End}(\cH)$, an orthogonal projector
$\Pi_{\rho,\delta}^{(n)}\in\operatorname{End}(\cH^{\otimes n})$
satisfies, for fixed $\delta>0$,
\[
 \Tr\!\left[\rho^{\otimes n}\Pi_{\rho,\delta}^{(n)}\right]=1-o_n(1)
\]
and
\[
 \e^{-n(H(\rho)+\delta)}\Pi_{\rho,\delta}^{(n)}
 \leq
 \Pi_{\rho,\delta}^{(n)}\rho^{\otimes n}\Pi_{\rho,\delta}^{(n)}
 \leq
 \e^{-n(H(\rho)-\delta)}\Pi_{\rho,\delta}^{(n)};
\]
see
\cite{Schumacher,JozsaSchumacher,NielsenChuang,Wilde,Tomamichel}.

There is a simple classical reason to expect these estimates to hold for all
marginals at once.  Let $p$ now be a joint distribution on
$\mathcal X_1\times\cdots\times\mathcal X_m$, and describe a word by its
\emph{empirical distribution}, or \emph{type}, $\widehat p_{x^n}$, which is a
probability vector on the same product alphabet.  The set of probability
vectors $r$ on that alphabet satisfying $\norm{r-p}_1\leq\delta$ is convex, and
the map taking a joint distribution to any marginal is affine and
$\ell^1$-contractive.  Hence
\[
 \norm{\widehat p_{x_S^n}-p_S}_1
 \leq\norm{\widehat p_{x^n}-p}_1,
 \qquad \varnothing\ne S\subseteq[m].
\]
In particular, one jointly typical word is marginally typical for every
$S$ simultaneously.  Put
$\mathcal T_n:=\{x^n\in(\mathcal X_1\times\cdots\times\mathcal X_m)^n:
\norm{\widehat p_{x^n}-p}_1\leq\delta\}$, and let
$q_n$ be obtained by conditioning $p^{\otimes n}$ on
$\mathcal T_n$.  The law of large numbers gives
$p^{\otimes n}(\mathcal T_n)\to1$, and
\[
 \norm{q_n-p^{\otimes n}}_1
 =2\bigl(1-p^{\otimes n}(\mathcal T_n)\bigr)\longrightarrow0.
\]
For every $S$, the marginal $(q_n)_S$ is supported on marginally typical
words and, for such a word $x_S^n$,
\[
 (q_n)_S(x_S^n)
 \leq\frac{p_S^{\otimes n}(x_S^n)}
 {p^{\otimes n}(\mathcal T_n)}
 \leq
 \frac{\e^{-n(H(p_S)-\gamma(\delta))}}
 {p^{\otimes n}(\mathcal T_n)}
 =\e^{-n(H(p_S)-\gamma(\delta)-o_n(1))},
\]
where $\gamma(\delta)\to0$ as $\delta\downarrow0$.  Indeed, on
the support of each $p_S$, the cross-entropy against $p_S$ is a continuous
linear function of the empirical distribution.  Words containing a symbol
outside $\supp p_S$ have probability zero under both $p_S^{\otimes n}$
and $(q_n)_S$.
There are only finitely many marginals.
Thus one conditioning operation gives the required max-probability, and
hence collision-probability, exponent for every marginal.  This is the
classical simultaneous-smoothing mechanism behind many multi-user coding
arguments.

The corresponding quantum statement is both natural and nontrivial.  It
asks for one nearby global state whose every marginal has the entropy scale
predicted by the original state.  Such a state would provide a common
typical object for multiparty state merging, randomness extraction,
decoupling, channel simulation, and coding for quantum networks.  But for
overlapping subsystems the lifted typical projectors need not commute, there
is no literal joint event on which to condition, and applying one projection
can destroy the typicality created by another.  The classical affine and
convex argument therefore has no immediate quantum analogue.  For the
operational origin of state merging, see \cite{HorodeckiOppenheimWinter}.

The multiparty typicality conjecture was formulated in
\cite[Conjecture~3.2.7]{DutilThesis}; see also
\cite{DutilHayden,Noetzel,DrescherFawzi}.  The conjecture predicts that one
can make a single small perturbation of $\rho^{\otimes n}$ that is typical
for every marginal at once.

\Needspace{8\baselineskip}
\begin{informalconjecture}[cf.\ \cref{conj:purity,conj:op} for details]
Given a multipartite state $\rho$, many copies $\rho^{\otimes n}$ can be
changed by an arbitrarily small amount in trace norm so that, for every
nonempty collection of parties $S$, the purity of the $S$-marginal is at
most $\e^{-n(H(\rho_S)-\delta)}$.  Equivalently, all marginal collision
entropies per copy can simultaneously be made at least their von Neumann
entropy rates, up to an arbitrarily small fixed slack $\delta$.  A stronger
version asks for the analogous bound on the largest eigenvalue of every
marginal.
\end{informalconjecture}

Together with the positive cases for fewer parties, the following result
gives the exact thresholds and maximal trace-distance separation.

\Needspace{8\baselineskip}
\begin{informalresult}[Main result; cf.\ \cref{thm:main} for details]
There is a unit vector $\Omega\in(\mathbb C^2)^{\otimes4}$ and a fixed
$\delta_0>0$ such that, for every fixed $0<\delta\leq\delta_0$, any
sequence of states satisfying the conjectured purity bounds with slack
$\delta$ for only the three overlapping marginals $AB$, $AC$, and
$BC$ becomes asymptotically perfectly distinguishable from
$\proj{\Omega}^{\otimes n}$: its trace distance from $\proj{\Omega}^{\otimes n}$
tends to the maximum value $2$.  Consequently the purity conjecture, its
stronger operator-norm version, and the pure-vector versions are all false.
Moreover, the mixed tripartite state $\Tr_R\proj{\Omega}$ is itself a
counterexample, and the one-shot simultaneous min-entropy conjectures fail as
well.
\end{informalresult}

The failure is stable under small single-copy perturbations.
\Cref{prop:robustness} shows that an open trace-norm neighborhood of
$\proj{\Omega}$ consists entirely of counterexamples, with a common
positive entropy slack and the same limiting trace distance $2$.
The neighborhood is taken in the full four-qubit state space, so it
includes nearby pure states, full-rank noisy states, and states outside
the singlet space.

The proof exhibits a mechanism for the failure of simultaneous typicality:
a linear entropy deficit competes with a quadratic cost of polynomial
concentration.  After symmetry reduction, the marginal purity bounds force
exponentially small weight in three frequency caps, with an exponent
linear in a small tolerance $\eta$.  The associated cap projections
nevertheless form a frame with a lower bound whose exponent is only
quadratic in $\eta$.  The latter estimate comes from representing
symmetric tensors by homogeneous polynomials and controlling concentration
in two small holes of $\CP^1$ through a Bernstein--Markov inequality and
a subharmonic barrier.  For sufficiently
small tolerance and entropy slack, the comparison forces the overlap with
the tensor-power input to vanish.  This brings polynomial approximation
and potential theory into the study of incompatible marginal entropy
constraints.

Let us review a few historical developments related to quantum typicality.
Quantum typical subspaces and noiseless quantum
coding were developed in the 1990s, notably by Schumacher and by
Jozsa--Schumacher \cite{Schumacher,JozsaSchumacher}.  Against this
background, Dutil's 2011 thesis isolated the multiparty typicality conjecture in the study of
assisted entanglement distillation and multiparty state merging
\cite{DutilThesis,DutilHayden}.  Dutil also proved the two-party case.

N\"otzel later gave a representation-theoretic proof of the two-party case
and formulated Schatten-$p$ extensions \cite{Noetzel}; our
mixed-tripartite example refutes those extensions for every fixed $p\geq2$.
His analysis made
explicit why the direct three-party argument runs into noncommuting
projectors for overlapping marginals.  In a one-shot setting, Drescher and
Fawzi formulated a simultaneous min-entropy-smoothing conjecture and proved
it for two quantum parties.  They also discuss commuting cases, state the
result for non-overlapping families, and note its consequence for pure
tripartite states \cite[Secs.~V--VI]{DrescherFawzi}.  The pure-tripartite
statement is proved in Drescher's report \cite[Corollary~5.5]{DrescherReport};
we give a direct proof of the i.i.d.\ result in \cref{prop:pure-tripartite}.
In the classical case, Drescher and Fawzi also establish sharp dependence of
the smoothing error on the number of marginals
\cite[Theorem~V.2]{DrescherFawzi}; the general optimality proof appears in
\cite[Theorem~4.3]{DrescherReport}.  For a fixed number of parties, this
error still tends to zero with the smoothing parameter.
The related difficulty of simultaneous decoding for quantum interference
channels was highlighted by Fawzi, Hayden, Savov, Sen and Wilde in work
first circulated in 2011 \cite{FawziEtAlInterference}.
Related work on one-shot smoothing and entropy inequalities includes
\cite{AnshuBertaJainTomamichel,RegulaLamiDatta}.
Together these results
identify genuine overlap, rather than merely the number of marginals, as the
essential difficulty.

Other joint-typicality techniques solved important operational problems
without settling simultaneous smoothing itself.  Sen developed one-shot
quantum joint-typicality, tilting, and augmentation methods for simultaneous
decoding \cite{SenJointTypicality,SenInnerBounds}, and Ding, Gharibyan,
Hayden, and Walter used related ideas in a quantum multiparty packing lemma
\cite{DingEtAl}.  More recently, Colomer and Winter obtained general
multi-user decoupling results without assuming simultaneous smoothing, while
still describing the conjecture as unresolved \cite{ColomerWinter};
multipartite convex splitting provides another route around overlapping
joint-typicality questions in channel simulation \cite{BertaChengGao}.
The counterexample identifies a limitation of a proposed proof method,
without contradicting operational rate regions established independently.
For example, Dutil's direct multiparty merging analysis beyond two senders
uses the typicality conjecture \cite[Sec.~3.2.4]{DutilThesis}, while
Dutil and Hayden conjecture independently smoothed one-shot merging costs
\cite[Conjecture~12]{DutilHayden}.  Chakraborty, Nema and Sen's passage
from their one-shot quantum multiple-access bound to the optimal i.i.d.\
region likewise assumes simultaneous smoothing
\cite[Sec.~1]{ChakrabortyNemaSen}.  Colomer and Winter prove merging and
multiple-access achievability results without that assumption
\cite[Sec.~6]{ColomerWinter}; Sen also obtains fully smoothed multipartite
covering and decoupling bounds by telescoping \cite{SenTelescoping}.
Thus these operational conclusions can survive even though the general
simultaneous-smoothing route to them is unavailable.

\subsection{Organization}

\Cref{sec:definitions} states the conjecture and its variants, formulates the
main counterexample theorem, derives its mixed-tripartite and one-shot
consequences, and identifies the party and dimension thresholds.
\Cref{sec:candidate} constructs the four-qubit state, identifies it with the
Higuchi--Sudbery state, and computes its cyclic recoupling and pair spectra.
\Cref{sec:reduction} reduces the problem to three frequency projections on
$\Sym^n(\mathbb C^2)$.  \Cref{sec:entropy} develops the probability
geometry and triangular containment used in the proof; its entropy-certainty
inequality provides geometric motivation.
\Cref{sec:toeplitz} combines coherent states, an elementary Bernstein--Markov
inequality, the triangle barrier, and binomial leakage to prove the frame
bound.  \Cref{sec:contradiction} completes the proof with explicit positive
entropy slack and extends it to larger systems.  \Cref{sec:discussion}
proves robustness in the full state space and explains why the quadratic
order of the frame exponent is optimal.  \Cref{app:matrices}
gives exact computational-basis matrices.

\section{Definitions and preliminaries}
\label{sec:definitions}

\subsection{Finite-dimensional quantum systems}

All quantum systems have finite-dimensional complex Hilbert spaces.
For such a space $E$, $\operatorname{End}(E)$ denotes its algebra of
\emph{complex-linear endomorphisms}, $I_E$ its \emph{identity operator}, and $X^*$
the \emph{Hilbert-space adjoint} of an operator $X$.  We write $\Tr$ for the
ordinary, \emph{unnormalized trace}.  A \emph{quantum state}, or
\emph{density operator}, on $\cH$ is a positive semidefinite operator
$\rho\in\operatorname{End}(\cH)$ with $\Tr\rho=1$.
A \emph{subnormalized state} is a positive operator in the same algebra with
trace at most one.  A \emph{pure state} is a rank-one state
$\rho=\proj{\psi}$, where $\psi\in\cH$ is a unit vector.  We freely
identify a named quantum system, such as $A$, with its Hilbert space
$\cH_A$; juxtaposition of system labels denotes their tensor product.
Inner products are antilinear in the first argument.  For vectors
$x_j,y_j\in E_j$ in Hilbert spaces $E_j$, the \emph{tensor-product inner
product} is determined by
\[
 \langle x_1\otimes\cdots\otimes x_n,
         y_1\otimes\cdots\otimes y_n\rangle
 =\prod_{j=1}^n\langle x_j,y_j\rangle.
\]
For a Hilbert space $E$, the symmetric group $S_n$ acts unitarily on
$E^{\otimes n}$ by
\[
 U_\pi^E(x_1\otimes\cdots\otimes x_n)
 =x_{\pi^{-1}(1)}\otimes\cdots\otimes x_{\pi^{-1}(n)},
 \qquad \pi\in S_n,\quad x_j\in E.
\]
We realize the \emph{symmetric power} as the invariant subspace
\[
 \Sym^n(E)=\{\xi\in E^{\otimes n}:U_\pi^E\xi=\xi
                       \text{ for every }\pi\in S_n\},
\]
equipped with the restriction of the tensor-product inner product.

For $m$ systems $A_1,\ldots,A_m$, write
\[
 A_{[m]}=A_1\otimes\cdots\otimes A_m,\qquad [m]=\{1,\ldots,m\}.
\]
If $T\subseteq[m]$, then
\[
 A_T=\bigotimes_{i\in T}A_i,\qquad T^c=[m]\setminus T.
\]
The \emph{partial trace} is a linear map
\[
 \Tr_{T^c}:\operatorname{End}(A_{[m]})\longrightarrow
           \operatorname{End}(A_T).
\]
The \emph{reduced state}, or \emph{marginal}, of $\rho\in\operatorname{End}(A_{[m]})$
on $A_T$ is
\[
 \rho_T=\Tr_{T^c}\rho\in\operatorname{End}(A_T).
\]
The partial trace is characterized by
\[
 \Tr\!\left[X_T\,\Tr_{T^c}Y\right]
 =\Tr\!\left[(X_T\otimes\1_{T^c})Y\right]
\]
for $X_T\in\operatorname{End}(A_T)$ and
$Y\in\operatorname{End}(A_{[m]})$.  For a unit vector
$\psi\in A_{[m]}$, we also write
$\psi_T:=\Tr_{T^c}\proj{\psi}\in\operatorname{End}(A_T)$
for the marginal of its pure state.

Our subscript records the systems that are \emph{retained}.
For $n$ independent copies, it is convenient to reorder tensor factors as
\[
 (A_1\otimes\cdots\otimes A_m)^{\otimes n}
 \cong A_1^{\otimes n}\otimes\cdots\otimes A_m^{\otimes n}.
\]
For $\sigma\in\operatorname{End}(A_{[m]}^{\otimes n})$, we denote the
corresponding marginal by
\[
 \sigma_{T^n}:=\Tr_{(T^c)^n}\sigma
              \in\operatorname{End}(A_T^{\otimes n}).
\]

For $X\in\operatorname{End}(E)$, its \emph{trace norm},
\emph{Hilbert--Schmidt norm}, and \emph{operator norm} are, respectively,
\[
 \norm{X}_1=\Tr\sqrt{X^*X},\qquad
 \norm{X}_2=(\Tr X^*X)^{1/2},\qquad
 \norm{X}_\infty=\sup_{\xi\in E,\,\norm{\xi}=1}\norm{X\xi}.
\]
For $X\geq0$, $\norm{X}_\infty$ is its largest eigenvalue.
More generally, its \emph{Schatten norm} for $1\leq p<\infty$ is
$\norm X_p=(\Tr[(X^*X)^{p/2}])^{1/p}$.
The \emph{von Neumann entropy}, \emph{collision entropy}, and
\emph{min-entropy} of a state are
\[
\begin{aligned}
 H(\rho)&=-\Tr(\rho\log \rho),\\
 H_2(\rho)&=-\log \Tr(\rho^2),\\
 H_{\min}(\rho)&=-\log \norm{\rho}_\infty.
\end{aligned}
\]
The quantity $\Tr(\rho^2)$ is its \emph{purity}.  The elementary inequality
\begin{equation}
 \Tr(\rho^2)\leq\norm{\rho}_\infty
\label{eq:purity-operator}
\end{equation}
shows that an operator-norm estimate is stronger than the corresponding
purity estimate.

We use the unnormalized \emph{trace-norm distance} $\norm{\rho-\sigma}_1$.
If $\rho=\proj{\psi}$ is pure, then
\begin{equation}
 \norm{\sigma-\proj{\psi}}_1\leq\eps
 \quad\Longrightarrow\quad
 \langle\psi,\sigma\psi\rangle\geq1-\frac{\eps}{2}.
\label{eq:trace-overlap}
\end{equation}
This follows by applying the two-outcome measurement
$\{\proj{\psi},\1-\proj{\psi}\}$.  For unit vectors,
\begin{equation}
 \norm{\proj{\psi}-\proj{\phi}}_1
 =2\sqrt{1-\abs{\langle\psi,\phi\rangle}^2}
 \leq2\norm{\psi-\phi}.
\label{eq:pure-state-distance}
\end{equation}
Thus a counterexample allowing arbitrary density-operator smoothings is
automatically a counterexample to a formulation allowing only nearby pure
vectors in Hilbert norm.
The precise relation to phase-optimized vector distance is
\begin{equation}
 \min_{\theta\in\mathbb R}\norm{\psi-\e^{i\theta}\phi}^2
 =2-2|\langle\psi,\phi\rangle|.
\label{eq:phase-distance}
\end{equation}

For subnormalized states $\rho,\sigma$ on the same Hilbert space, define
their \emph{root fidelity}, \emph{generalized root fidelity}, and
\emph{purified distance} by
\begin{equation}
\begin{aligned}
 F(\rho,\sigma)&=\norm{\sqrt\rho\sqrt\sigma}_1,\\
 \bar F(\rho,\sigma)&=F(\rho,\sigma)
       +\sqrt{(1-\Tr\rho)(1-\Tr\sigma)},\\
 \mathsf P(\rho,\sigma)&=\sqrt{1-\bar F(\rho,\sigma)^2}.
\end{aligned}
\label{eq:fidelity-distance}
\end{equation}
These quantities lie in $[0,1]$, and $\bar F=F$ if either state is
normalized.  Hilbert--Schmidt Cauchy--Schwarz, applied in the variational
formula for the trace norm, gives
$F(\rho,\sigma)^2\leq\Tr\rho\,\Tr\sigma$.
The Fuchs--van de Graaf inequalities for normalized states and their
subnormalized extension imply
\begin{equation}
 1-F(\rho,\sigma)\leq\tfrac12\norm{\rho-\sigma}_1
 \quad(\Tr\rho=\Tr\sigma=1),\qquad
 \tfrac12\norm{\rho-\sigma}_1\leq\mathsf P(\rho,\sigma).
\label{eq:fidelity-trace-bounds}
\end{equation}
See \cite{FuchsVanDeGraaf} and \cite[Chap.~3]{Tomamichel}.
For $0<\alpha<1$, the unconditioned \emph{smooth min-entropy} is
\[
 H_{\min}^{\alpha}(\rho)
 :=\sup_{\substack{\tau\geq0,\ \Tr\tau\leq1\\
                   \mathsf P(\tau,\rho)\leq\alpha}}
      -\log\norm\tau_\infty.
\]
For a marginal $\rho_T$, we also write
$H_{\min}^{\alpha}(T)_\rho=H_{\min}^{\alpha}(\rho_T)$.

\subsection{The conjecture and main result}

We start by stating the conjecture that motivated this paper, with its
entropy and tolerance expressed in the natural-logarithm convention.

\begin{conjecture}[Multiparty quantum typicality]
\label[conjecture]{conj:purity}
Let $\rho\in\operatorname{End}(A_{[m]})$ be a state.  For every
$\eps>0$ and $\delta>0$, and all sufficiently large $n$, there is a
state $\sigma_n\in\operatorname{End}(A_{[m]}^{\otimes n})$ such that
\begin{equation}
 \norm{\sigma_n-\rho^{\otimes n}}_1\leq\eps
\label{eq:typicality-distance}
\end{equation}
and, for every nonempty $T\subseteq[m]$,
\begin{equation}
 \Tr\!\left[(\sigma_n)_{T^n}^2\right]
 \leq \e^{-n(H(\rho_T)-\delta)}.
\label{eq:typicality-purity}
\end{equation}
\end{conjecture}

For $m=1$, this follows from the usual typical-subspace theorem.  For
$m=2$, it was proved by Dutil and later by N\"otzel
\cite{DutilThesis,Noetzel}.  For mixed inputs, the conjecture was open for
general overlapping families when $m\geq3$; the pure tripartite case was
established in \cite[Sec.~VI.B]{DrescherFawzi} and
\cite[Corollary~5.5]{DrescherReport}.

\begin{remark}
Dutil's original statement, written here in natural-log form, includes a
constant prefactor.  For every fixed $\eps'>0$ and every choice of
positive marginal slacks $\delta_T$, it asks, for all sufficiently large
$n$, for a state satisfying
\[
 \Tr\!\left[(\sigma_n)_{T^n}^2\right]
 \leq(1-\mu(\eps'))^{-2}
 \e^{-n(H(\rho_T)-\delta_T)},
\]
where $\mu(\eps')\to0$, while its trace-distance bound is
$\nu(\eps')$ with $\nu(\eps')\to0$
\cite[Conjecture~3.2.7]{DutilThesis}.
Both the entropy and $\delta_T$ here are $\log2$ times their binary-log
counterparts in that source.  This formulation implies \cref{conj:purity}:
given $\eps,\delta>0$, choose $\eps'$ with
$\nu(\eps')\leq\eps$ and $\mu(\eps')<1$, set every
$\delta_T=\delta/2$, and absorb the fixed prefactor into
$\e^{n\delta/2}$.  Refuting \cref{conj:purity} therefore also refutes
Dutil's original formulation.
\end{remark}

The following stronger formulation is the i.i.d.\ consequence of the
one-shot min-entropy conjecture of Drescher and Fawzi, as they note in
\cite[Sec.~III]{DrescherFawzi}.

\begin{conjecture}[Simultaneous operator-norm typicality]
\label[conjecture]{conj:op}
Under the hypotheses of \cref{conj:purity}, one can choose $\sigma_n$
such that
\begin{equation}
 \norm{(\sigma_n)_{T^n}}_\infty
 \leq\e^{-n(H(\rho_T)-\delta)}
\label{eq:typicality-operator}
\end{equation}
for every nonempty $T\subseteq[m]$.
\end{conjecture}

\Cref{conj:op} implies \cref{conj:purity}, since
$\Tr(\tau^2)\leq\norm\tau_\infty$ for every density operator $\tau$
by \cref{eq:purity-operator}.

It is therefore enough to refute \cref{conj:purity}, and we work with that
purity formulation throughout the paper.  The
logarithmic form of \cref{eq:typicality-operator} is
\begin{equation}
 -\frac1n\log \norm{(\sigma_n)_{T^n}}_\infty
 \geq H(\rho_T)-\delta.
\label{eq:typicality-min-entropy}
\end{equation}

\begin{definition}[Pure-vector variants]
\label[definition]{def:pure}
For a pure input $\rho=\proj{\psi}$ with $\psi\in A_{[m]}$, the
\emph{pure-vector variants} of
\cref{conj:purity,conj:op} require the smoothing to have the form
$\sigma_n=\proj{\psi_n}$, where $\psi_n\in A_{[m]}^{\otimes n}$
is a unit vector and, after a choice of global phase,
\[
 \norm{\psi_n-\psi^{\otimes n}}\leq\eps.
\]
The marginal purity bounds in \cref{eq:typicality-purity}, or respectively the operator-norm
bounds in \cref{eq:typicality-operator}, are otherwise unchanged.
\end{definition}

\Needspace{8\baselineskip}
We can now state the main result.

\begin{theorem}[Counterexample]
\label[theorem]{thm:main}
The unit vector $\Omega\in(\mathbb C^2)^{\otimes4}=
\cH_A\otimes\cH_B\otimes\cH_C\otimes\cH_R
$ defined in \cref{eq:target-vector} has common pair entropy
\[
 h:=H(\Omega_{AB})=H(\Omega_{AC})=H(\Omega_{BC})
   =\log2+\tfrac12\log3.
\]
There exist $\delta_0\in(0,h)$, $r>0$, and $n_0\in\mathbb N$
such that, for every $n\geq n_0$, every $0<\delta\leq\delta_0$, and
every density operator
$\sigma_n\in\operatorname{End}\!\left(
 ((\mathbb C^2)^{\otimes4})^{\otimes n}\right)$ satisfying the three
purity bounds
\begin{equation}
 \Tr\!\left[(\sigma_n)_{(AB)^n}^2\right],
 \ \Tr\!\left[(\sigma_n)_{(AC)^n}^2\right],
 \ \Tr\!\left[(\sigma_n)_{(BC)^n}^2\right]
 \leq\e^{-n(h-\delta)},
\label{eq:main-purity}
\end{equation}
one has
\begin{equation}
 \langle\Omega^{\otimes n},\sigma_n\Omega^{\otimes n}\rangle
 \leq3\sqrt{n+1}\,\e^{-rn}.
\label{eq:main-quantitative}
\end{equation}
Consequently, every sequence satisfying these bounds for all sufficiently
large $n$ obeys
\begin{equation}
 \langle\Omega^{\otimes n},\sigma_n\Omega^{\otimes n}\rangle
 \xrightarrow[n\to\infty]{}0,
 \qquad\text{hence}\qquad
 \norm{\sigma_n-\proj{\Omega}^{\otimes n}}_1
 \xrightarrow[n\to\infty]{}2.
\label{eq:main-limits}
\end{equation}
One may take $\delta_0=\tfrac12[h-f(1/2-10^{-4})]$, where
$f:[0,1]\to\mathbb R$ is given by
$f(u)=-u\log u-(1-u)\log(1-u)+u\log3$, with $0\log0=0$.  This gives
$\delta_0\simeq5.49406\cdot10^{-5}$.  One may then take
$r=\delta_0/2-140\cdot10^{-8}\simeq2.60703\cdot10^{-5}$.
The threshold $n_0$ is independent of $\delta$ and of the candidate state.
\end{theorem}
Note that the same conclusion holds, a fortiori, if the purities in \cref{eq:main-purity} are
replaced by operator norms.

Consequently, \cref{conj:purity,conj:op} are false already for four
qubit parties: they fail for this $\Omega$ with every $0<\delta\leq\delta_0$ and
with \emph{every} $\eps<2$.  The pure-vector variants of
\cref{def:pure} fail for every $\eps<\sqrt2$.  This is their entire
nontrivial tolerance range, since the squared distance in
\cref{eq:phase-distance} is always at most $2$.

\begin{remark}
Two is the largest possible value of $\norm{\cdot}_1$ between two states, so
\cref{eq:main-limits} is the strongest failure the definitions allow: a state meeting the
three purity constraints is not merely far from $\proj{\Omega}^{\otimes n}$, it is
asymptotically perfectly distinguishable from it by a single measurement.  In
particular, the conclusion holds for every fixed trace-distance tolerance
strictly below $2$.
\end{remark}

\begin{corollary}[Mixed tripartite counterexample]
\label[corollary]{cor:mixed-tripartite}
Let $\rho_{ABC}=\Tr_R\proj{\Omega}\in\operatorname{End}(ABC)$.
Every sequence of states
$\sigma_n\in\operatorname{End}((ABC)^{\otimes n})$ satisfying
\[
 \Tr[(\sigma_n)_{(XY)^n}^2]\leq\e^{-n(h-\delta_0)},
 \qquad XY\in\{AB,AC,BC\},
\]
for all sufficiently large $n$ obeys
$\norm{\sigma_n-\rho_{ABC}^{\otimes n}}_1\to2$.
Thus both simultaneous typicality conjectures fail already for a mixed
state of three qubits, with maximal trace-distance separation.
\end{corollary}

\begin{proof}
Uhlmann's theorem characterizes root fidelity as the maximal absolute
overlap of purifications \cite[Theorem~9.4]{NielsenChuang}.
Thus one can choose a finite-dimensional ancillary space
$E$ and a unit-vector purification
$\Psi_n\in(ABCR)^{\otimes n}\otimes E$ of $\sigma_n$ with
\[
 |\langle\Omega^{\otimes n}\otimes e,\Psi_n\rangle|
 =F(\rho_{ABC}^{\otimes n},\sigma_n)
\]
for a fixed unit vector $e\in E$.
The state $\widehat\sigma_n=\Tr_E\proj{\Psi_n}
\in\operatorname{End}((ABCR)^{\otimes n})$ has the same three
pair marginals as $\sigma_n$, and
\[
 \langle\Omega^{\otimes n},\widehat\sigma_n\Omega^{\otimes n}\rangle
 \geq F(\rho_{ABC}^{\otimes n},\sigma_n)^2.
\]
The left side tends to zero by \cref{thm:main}, so
$F(\rho_{ABC}^{\otimes n},\sigma_n)\to0$.  The Fuchs--van de Graaf
inequality \cite{FuchsVanDeGraaf}, recalled in
\cref{eq:fidelity-trace-bounds}, gives
\[
 2\geq\norm{\sigma_n-\rho_{ABC}^{\otimes n}}_1
 \geq2\bigl(1-F(\rho_{ABC}^{\otimes n},\sigma_n)\bigr)
 \longrightarrow2.
\]
No normalization or additional entropy slack is needed in this
extension, since \cref{thm:main} allows mixed competitors.
\end{proof}

This also refutes the Schatten-norm extension in
\cite[Conjecture~2]{Noetzel}, and indeed its version for every fixed real
$p\geq2$.  For a density operator $\tau$, H\"older interpolation gives
\[
 \Tr\tau^2\leq(\Tr\tau^p)^{1/(p-1)},\qquad 2\leq p<\infty.
\]
Thus a bound $\Tr\tau^p\leq\e^{-n(p-1)(H-\delta)}$ implies
$\Tr\tau^2\leq\e^{-n(H-\delta)}$.  The case $p=\infty$ is
\cref{eq:purity-operator}; no assertion for $1<p<2$ follows from this implication.

In unconditioned form, both one-shot conjectures assert the following.
For each number of parties $m$, there are
dimension-independent functions $g_m(\alpha)\to0$ as $\alpha\downarrow0$
and $0\leq h_m(\alpha)<\infty$ such that, for every normalized multipartite
state $\rho$ and every $0<\alpha<1$, there is a subnormalized state
$\sigma$ satisfying
\begin{equation}
 \mathsf P(\sigma,\rho)\leq g_m(\alpha),\qquad
 H_{\min}(\sigma_T)\geq H_{\min}^{\alpha}(T)_\rho-h_m(\alpha)
 \quad(\emptyset\ne T\subseteq[m]).
\label{eq:one-shot-demand}
\end{equation}
Drescher--Fawzi ask for $h_m=0$
\cite[Conjecture~III.1]{DrescherFawzi}.  Taking the conditioning system
to be trivial in Colomer--Winter gives \cref{eq:one-shot-demand} with a
finite additive loss \cite[Eq.~(2.1)]{ColomerWinter}.
Allowing subnormalized competitors weakens a formulation that requires
normalized ones, so refuting this common statement suffices for both.

\begin{corollary}[Failure of one-shot simultaneous smoothing]
\label[corollary]{cor:one-shot}
The simultaneous min-entropy-smoothing conjecture of
\cite[Conjecture~III.1]{DrescherFawzi} and its conditional variant
\cite[Eq.~(2.1)]{ColomerWinter} are false in general.  Failure already occurs
for three unconditioned parties, each a tensor power of a qubit system.
\end{corollary}

\begin{proof}
Let $\rho=\rho_{ABC}$ be the state of \cref{cor:mixed-tripartite}, and
suppose that \cref{eq:one-shot-demand} holds for $m=3$.
Fix $0<\alpha<1$ with $g_3(\alpha)<1$.  The fixed-error lower bound
in the quantum asymptotic equipartition property
\cite{TomamichelColbeckRenner}, in the purified-distance formulation of
\cite[Sec.~6.4]{Tomamichel}, gives, for the three pairs and all sufficiently
large $n$,
\[
 H_{\min}^{\alpha}((XY)^n)_{\rho^{\otimes n}}
 \geq n(h-\delta_0/4).
\]
For these unconditioned marginals, the displayed bound also follows by
restricting to an ordinary spectral typical subspace: its mass tends to one,
and its largest retained eigenvalue has the required exponential rate.
A conjectured simultaneous smoothing
$\sigma_n\in\operatorname{End}((ABC)^{\otimes n})$ would satisfy
\[
 \mathsf P(\sigma_n,\rho^{\otimes n})\leq g_3(\alpha),\qquad
 \norm{(\sigma_n)_{(XY)^n}}_\infty
 \leq\exp[-n(h-\delta_0/4)+h_3(\alpha)].
\]
The original conjecture permits subnormalized states.  Write
$t_n=\Tr\sigma_n\leq1$.  Since the target is normalized, its
generalized fidelity with $\sigma_n$ is the usual root fidelity, and
\[
 t_n\geq F(\rho^{\otimes n},\sigma_n)^2
 \geq1-g_3(\alpha)^2>0.
\]
Normalization $\widetilde\sigma_n=\sigma_n/t_n$ increases fidelity and
therefore does not increase purified distance from the normalized target.
It changes the logarithmic marginal bounds by at most the fixed number
$-\log(1-g_3(\alpha)^2)$.  This number and $h_3(\alpha)$ are absorbed
into $n\delta_0/4$.  Thus for all sufficiently large $n$,
\[
 \norm{(\widetilde\sigma_n)_{(XY)^n}}_\infty
 \leq\e^{-n(h-\delta_0/2)}.
\]
These bounds imply the purity hypotheses of \cref{cor:mixed-tripartite},
so $\norm{\widetilde\sigma_n-\rho^{\otimes n}}_1\to2$.  But
$\tfrac12\norm{\widetilde\sigma_n-\rho^{\otimes n}}_1
\leq\mathsf P(\widetilde\sigma_n,\rho^{\otimes n})\leq g_3(\alpha)<1$
by \cref{eq:fidelity-trace-bounds}, a contradiction.
The number of parties stays three while their dimensions
grow with $n$; this is precisely where the dimension independence of
$g_m$ and $h_m$ is used.
\end{proof}

Drescher and Fawzi state the one-shot result for non-overlapping families,
omitting the proof for reasons of space, and note the pure-tripartite consequence
\cite[Sec.~VI.B]{DrescherFawzi}.  The proofs are given in Drescher's report
\cite[Theorem~5.4 and Corollary~5.5]{DrescherReport}.
The i.i.d.\ statement needed for the party threshold has the following direct proof.

\begin{proposition}[Pure tripartite typicality]
\label[proposition]{prop:pure-tripartite}
For every pure tripartite input, \cref{conj:op,conj:purity} and their
pure-vector variants hold.
\end{proposition}

\begin{proof}
Let $\psi\in A\otimes B\otimes C$ be a unit vector and fix $\delta>0$.
For each $X\in\{A,B,C\}$, let $\Pi_X$ be the spectral typical
projection of $\psi_X^{\otimes n}$ with tolerance $\delta/2$.  Then
\[
 \Tr(\Pi_X\psi_X^{\otimes n})\longrightarrow1,\qquad
 \Pi_X\psi_X^{\otimes n}\Pi_X
 \leq\e^{-n(H(\psi_X)-\delta/2)}\Pi_X.
\]
The three lifted projections commute.  Their union bound shows that
\[
 p_n:=\norm{(\Pi_A\otimes\Pi_B\otimes\Pi_C)\psi^{\otimes n}}^2
 \longrightarrow1.
\]
For all sufficiently large $n$, define
\[
 \phi_n=p_n^{-1/2}(\Pi_A\otimes\Pi_B\otimes\Pi_C)\psi^{\otimes n},
 \qquad \sigma_n=\proj{\phi_n}.
\]
For a positive operator $Y$ on $A\otimes BC$ and a projection $Q$
on $BC$, cyclicity of the partial trace gives
\[
 \Tr_{BC}Y-\Tr_{BC}[(I_A\otimes Q)Y(I_A\otimes Q)]
 =\Tr_{BC}[(I_A\otimes(I_{BC}-Q))Y(I_A\otimes(I_{BC}-Q))]\geq0.
\]
Apply this on $A^{\otimes n}\otimes(BC)^{\otimes n}$ with
$Q=\Pi_B\otimes\Pi_C$, keeping $\Pi_A$ outside the partial trace.  It follows that
\[
 (\sigma_n)_{A^n}\leq p_n^{-1}\Pi_A\psi_A^{\otimes n}\Pi_A,
 \qquad
 \norm{(\sigma_n)_{A^n}}_\infty
 \leq p_n^{-1}\e^{-n(H(\psi_A)-\delta/2)}
 \leq\e^{-n(H(\psi_A)-\delta)}
\]
eventually.  The same argument applies to $B,C$.  Purity of both
$\sigma_n$ and $\proj{\psi}$ identifies the nonzero spectra of each
pair marginal with those of its complementary single marginal, giving
the required pair bounds.  The full-system bound is automatic since
$H(\proj{\psi})=0$.  Finally,
\[
 \norm{\sigma_n-\proj{\psi}^{\otimes n}}_1=2\sqrt{1-p_n}\longrightarrow0,
 \qquad
 \norm{\phi_n-\psi^{\otimes n}}^2=2-2\sqrt{p_n}\longrightarrow0.
\]
This proves operator-norm typicality and its pure-vector version; purity
bounds follow from \cref{eq:purity-operator}.
\end{proof}

\begin{remark}[Party and dimension thresholds]
\label[remark]{rem:pure-tripartite}
The known one- and two-party cases, \cref{prop:pure-tripartite}, and
\cref{thm:main} make four parties the exact threshold for failure on pure
inputs.  Together with the two-party positive result,
\cref{cor:mixed-tripartite} gives the threshold of three for mixed inputs.

The theorem extends to every local dimension $d\geq2$ and every number of
parties $m\geq4$; see \cref{cor:larger}.  Thus qubits already realize the
obstruction in the smallest nontrivial local dimension.
\end{remark}

\section{The counterexample}
\label{sec:candidate}

\subsection{Spin-one-half singlets}

Let $V=\mathbb C^2$ carry the fundamental representation of $SU(2)$,
with orthonormal \emph{computational basis}
\[
 \ket{0}=\begin{pmatrix}1\\0\end{pmatrix},\qquad
 \ket{1}=\begin{pmatrix}0\\1\end{pmatrix}.
\]
Write $U(g):V\to V$ for the defining unitary action of $g\in SU(2)$.
Its \emph{diagonal action} on $V^{\otimes n}$ is $U(g)^{\otimes n}$.
A nonzero vector $\psi\in V^{\otimes n}$ is called a \emph{singlet} if
\[
 U(g)^{\otimes n}\psi=\psi\qquad\text{for every }g\in SU(2).
\]
Equivalently, it has total spin zero and spans a copy of the
one-dimensional trivial representation.
The four physical systems are $A=B=C=R=V$, and their joint Hilbert space is
\[
 \cK=A\otimes B\otimes C\otimes R=V^{\otimes4}.
\]
The Clebsch--Gordan decomposition yields
\begin{equation}
 V\otimes V=V_0\oplus V_1,\qquad d_0=1,\quad d_1=3.
\label{eq:clebsch-gordan}
\end{equation}
Here $V_\ell\subset V\otimes V$ is the \emph{spin-$\ell$ subspace} and
$d_\ell=\dim V_\ell$.  In the computational basis, write
\[
 s=\frac{\ket{01}-\ket{10}}{\sqrt2},\qquad
 t_+=\ket{00},\quad t_0=\frac{\ket{01}+\ket{10}}{\sqrt2},\quad
 t_-=\ket{11}.
\]
Thus $s\in V_0$ and $t_+,t_0,t_-\in V_1$.
In the order $A,B,C,R$, define
\begin{equation}
 e_0=s_{AB}\otimes s_{CR},\qquad
 e_1=\frac{t_+^{AB}\otimes t_-^{CR}-t_0^{AB}\otimes t_0^{CR}
                 +t_-^{AB}\otimes t_+^{CR}}{\sqrt3}.
\label{eq:singlet-basis}
\end{equation}
These are orthonormal singlets and form a basis of the invariant subspace:
\begin{equation}
 \cS:=\Inv_{SU(2)}(\cK)=\operatorname{span}\{e_0,e_1\}
 \subset\cK,\qquad \dim\cS=2.
\label{eq:singlet-space}
\end{equation}
Thus the trivial representation occurs with multiplicity two in $\cK$;
$\cS$ is the corresponding \emph{singlet space}.
For details, see, for example, \cite{Varshalovich}.

We use a separate copy $\mathcal M:=\mathbb C^2$, called the
\emph{singlet multiplicity space}, for the coordinates in this basis.
Let $f_0,f_1$ be its standard orthonormal basis and define the isometry
\[
 J:\mathcal M\longrightarrow\cK,\qquad Jf_\ell=e_\ell,
 \qquad \ell=0,1.
\]
Then $\operatorname{Im}J=\cS$; viewed as a map onto $\cS$,
$J$ is unitary for the standard inner product on $\mathcal M$ and
the inherited inner product on $\cS$.

\subsection{Cyclic recoupling and the target vector}

Let $\gamma=(B\ C\ R)$ be the three-cycle fixing $A$, and let
$U_\gamma\in\operatorname{End}(\cK)$ be the unitary that
permutes its tensor factors:
\[
 U_\gamma(\ket{abcr})
 =\ket{arbc},\qquad
 a,b,c,r\in\{0,1\}.
\]

Substitution in \cref{eq:singlet-basis} gives
\[
 U_\gamma e_0=-\tfrac12e_0-\tfrac{\sqrt3}{2}e_1,\qquad
 U_\gamma e_1=\tfrac{\sqrt3}{2}e_0-\tfrac12e_1.
\]
Its restriction to $\cS$, expressed in the ordered basis $(e_0,e_1)$,
is represented on $\mathcal M$ by
\begin{equation}
 \Gamma=\begin{pmatrix}-\frac12&\frac{\sqrt3}{2}\\
                  -\frac{\sqrt3}{2}&-\frac12\end{pmatrix}
 \in\operatorname{End}(\mathcal M),
 \qquad \Gamma^T\Gamma=I,\quad \Gamma^3=I.
\label{eq:cyclic-matrix}
\end{equation}
The inverse cycle transposes this matrix and reverses the order of the
three bases, with no effect on the argument.  Define the scalar
$\omega\in\mathbb C$ and vectors $v,\bar v\in\mathcal M$ by
\begin{equation}
 \omega=\e^{2\pi i/3},\qquad
 v=\frac1{\sqrt2}\begin{pmatrix}1\\-i\end{pmatrix},\qquad
 \bar v=\frac1{\sqrt2}\begin{pmatrix}1\\i\end{pmatrix}.
\label{eq:cyclic-eigenvectors}
\end{equation}
Then $(v,\bar v)$ is orthonormal and
\begin{equation}
 \Gamma v=\bar\omega v,\qquad \Gamma\bar v=\omega\bar v.
\label{eq:cyclic-eigenvalues}
\end{equation}
Our candidate for a counterexample is the vector
\begin{equation}
 \boxed{\ \Omega=\frac{e_0-i e_1}{\sqrt2}=Jv\in\cS.\ }
\label{eq:target-vector}
\end{equation}

This is the Higuchi--Sudbery state \cite[Sec.~3]{HiguchiSudbery}, up to
an overall phase.  More explicitly, with the same tensor order,
\begin{equation}
 \Omega=\frac{-i}{\sqrt6}\bigl[
 \ket{0011}+\ket{1100}
 +\omega(\ket{1010}+\ket{0101})
 +\omega^2(\ket{1001}+\ket{0110})\bigr].
\label{eq:target-computational}
\end{equation}
Higuchi and Sudbery introduced this state as a candidate for maximizing
the average entanglement entropy across the three two-versus-two
bipartitions; they proved stationarity and gave numerical evidence for
maximality.  Their construction also identifies it through the symmetry
of the two-dimensional singlet space.  Brierley and Higuchi subsequently
proved local maximality \cite{BrierleyHiguchi}.  Gour and Wallach proved
maximality, uniquely up to local unitaries, within their family
$\mathcal A$ of superpositions of matching Bell-pair products
\cite[Theorem~10]{GourWallach}.  These results do not assert unrestricted
global maximality of the average von Neumann entropy.
For a recent account of the impossibility of four-qubit pure states with
all two-qubit marginals maximally mixed, see \cite{HuberSiewert}.
Exact computational-basis matrices for $\Omega$ and its marginals are
given in \cref{app:matrices}.

\subsection{Pair channels and entropy}

For a pair $XY$, let $\Pi_\ell^{XY}\in\operatorname{End}(X\otimes Y)$
be the orthogonal projector onto its spin-$\ell$ subspace.
The complementary spin spaces in $e_0,e_1$ are orthogonal, so partial
trace kills the off-diagonal singlet matrix elements.  Since $e_\ell$
is maximally entangled between the two spin-$\ell$ spaces,
\begin{equation}
 \Tr_{CR}\ketbra{e_\ell}{e_{\ell'}}
 =\delta_{\ell\ell'}\frac{\Pi_\ell^{AB}}{d_\ell}.
\label{eq:singlet-partial-trace}
\end{equation}
The isometry $J$ defines the embedding
\[
 \phi:\operatorname{End}(\mathcal M)\longrightarrow
       \operatorname{End}(\cK),\qquad \phi(X)=JXJ^*,
\]
which is trace preserving and nonunital: $\phi(I)=JJ^*=\Pi_{\mathrm{inv}}$, the
orthogonal projection in $\operatorname{End}(\cK)$ onto $\cS$.
For every pair $XY$, define its
\emph{marginal channel} on the multiplicity space by
\[
 \cN_{XY}:\operatorname{End}(\mathcal M)\longrightarrow
          \operatorname{End}(X\otimes Y),\qquad
 \cN_{XY}(M)=\Tr_{(XY)^c}(JMJ^*).
\]
In particular, \cref{eq:singlet-partial-trace} gives
\begin{equation}
 \cN_{AB}(X)=X_{00}\Pi_0^{AB}+X_{11}\frac{\Pi_1^{AB}}3.
\label{eq:ab-channel}
\end{equation}

\begin{lemma}[Cyclic covariance of the pair channels]
\label[lemma]{lem:cyclic-channels}
The isometry $J$ intertwines the cyclic actions:
\[
 U_\gamma J=J\Gamma,\qquad
 \phi(\Gamma X\Gamma^*)=U_\gamma\phi(X)U_\gamma^*.
\]
Writing $S_0=AB$, $S_1=AC$, and $S_2=BC$, one has, for every
$X\in\operatorname{End}(\mathcal M)$,
\begin{equation}
 \cN_{S_a}(X)=\sum_{\ell=0}^1
 (\Gamma^{-a}X\Gamma^a)_{\ell\ell}\frac{\Pi_\ell^{S_a}}{d_\ell},
 \qquad a=0,1,2.
\label{eq:cyclic-channels}
\end{equation}
Moreover, $\proj{\Omega}$ is fixed by conjugation with $U_\gamma$
and with every diagonal spin rotation $U(g)^{\otimes4}$.
\end{lemma}

\begin{proof}
The intertwining identity is the definition of the matrix $\Gamma$ in
\cref{eq:cyclic-matrix}, and the covariance of $\phi$ follows.  After relabeling the
output pair, the $AC$ marginal of $\phi(X)$ is the $AB$ marginal
of $U_\gamma^{-1}\phi(X)U_\gamma=\phi(\Gamma^{-1}X\Gamma)$.
Similarly, the second inverse cycle gives the $AR$ marginal and
the matrix $\Gamma^{-2}X\Gamma^2$.  In the $AR|BC$ coupling, each
singlet has the form $\ket{(j,j);0}$, $j=0,1$, and is maximally
entangled between the two spin-$j$ spaces.  Tracing either side of
$\ketbra{(j,j);0}{(j',j');0}$ gives zero if $j\ne j'$, because
the spin spaces on the traced side are orthogonal; if $j=j'$, it gives
$\Pi_j/d_j$ on the retained side.  The two marginals therefore have
the same weights $(\Gamma^{-2}X\Gamma^2)_{jj}$.
Thus the $BC$ channel has the same formula
with output projectors on $BC$.  Together with \cref{eq:ab-channel}, this proves
\cref{eq:cyclic-channels}.  Finally, $\Omega=Jv$ is a singlet and
$U_\gamma\Omega=J\Gamma v=\bar\omega\Omega$, proving both invariances.
\end{proof}

For the target, all three spin distributions are $p=(1/2,1/2)$, and
\begin{equation}
 \operatorname{spec}(\Omega_{XY})
 =\left(\tfrac12,\tfrac16,\tfrac16,\tfrac16\right),
 \qquad XY\in\{AB,AC,BC\}.
\label{eq:pair-spectra}
\end{equation}
In particular,
\begin{equation}
 h=\log2+\tfrac12\log3\simeq1.24245332489,
 \qquad \Tr(\Omega_{XY}^2)=\tfrac13.
\label{eq:pair-entropy-purity}
\end{equation}
The marginal collision entropy is $\log3<h$; the positive gap explains
why the unmodified tensor power does not already satisfy the desired
purity rates.

\section{Reduction to a three-projector frame}
\label{sec:reduction}

\subsection{Compression to the symmetric singlet space}

In this section, we reduce an arbitrary physical state
$\sigma_n\in\operatorname{End}(\cK^{\otimes n})$ to an operator
$\tau_n$ on an $(n+1)$-dimensional space.  The reduction preserves
the target overlap, and its pair-channel outputs inherit the required
marginal bounds.

For integers $n\geq1$, write
\[
 \mathcal H_n:=\Sym^n(\mathcal M)\subset\mathcal M^{\otimes n},
 \qquad N_n:=\dim\mathcal H_n=n+1,
\]
with the inherited tensor-product inner product of \cref{sec:definitions}.
The isometry
\[
 J_n:=J^{\otimes n}|_{\mathcal H_n}:
 \mathcal H_n\longrightarrow\cK^{\otimes n}
\]
has image $\Sym^n(\cS)$: every physical copy is a singlet, and the
vector is invariant under permutations of the copies.  These two
conditions are imposed by the orthogonal projections
$\Pi_{\mathrm{inv}}^{\otimes n}$ onto $\cS^{\otimes n}$ and
$\Pi_{\Sym}=\frac1{n!}\sum_{\pi\in S_n}U_\pi$ onto $\Sym^n(\cK)$,
where $U_\pi=U_\pi^{\cK}$ permutes the full copies.
Since every $U_\pi$ commutes with the identical tensor power
$\Pi_{\mathrm{inv}}^{\otimes n}$, the projections commute, and
\[
 P:=J_nJ_n^*=\Pi_{\Sym}\Pi_{\mathrm{inv}}^{\otimes n},
 \qquad J_n^*J_n=I_{\mathcal H_n}.
\]
Define the \emph{compression} of $\sigma_n$ by
\begin{equation}
 \tau_n:=J_n^*\sigma_nJ_n\in\operatorname{End}(\mathcal H_n),
 \qquad J_n\tau_nJ_n^*=P\sigma_nP.
\label{eq:compression}
\end{equation}
When applying $\cN_{XY}^{\otimes n}$ to an operator on $\mathcal H_n$,
we extend it by zero on $\mathcal H_n^\perp$ in
$\mathcal M^{\otimes n}$.

\begin{lemma}[Symmetric singlet compression]
\label[lemma]{lem:singlet-compression}
The operator $\tau_n$ in \cref{eq:compression} is positive, with
$\Tr\tau_n\leq1$, and preserves the target overlap:
\begin{equation}
 a_n:=\langle v^{\otimes n},\tau_nv^{\otimes n}\rangle
 =\langle\Omega^{\otimes n},\sigma_n\Omega^{\otimes n}\rangle,
 \qquad a_n\leq\Tr\tau_n.
\label{eq:compressed-overlap}
\end{equation}
For every pair $XY$ of distinct parties and $p\in[1,\infty]$,
\begin{equation}
 \norm{\cN_{XY}^{\otimes n}(\tau_n)}_p
 \leq\norm{(\sigma_n)_{(XY)^n}}_p.
\label{eq:compression-marginals}
\end{equation}
In particular, the pair-channel outputs inherit the purity and
operator-norm bounds of the original state.
\end{lemma}

\begin{proof}
\emph{Trace and overlap.}
Positivity follows from \cref{eq:compression}, and
$\Tr\tau_n=\Tr(P\sigma_n)\leq\Tr\sigma_n=1$.
Since $J_nv^{\otimes n}=\Omega^{\otimes n}$, the overlap identity
is immediate.  Positivity of $\tau_n$ and
$\norm{v^{\otimes n}}=1$ give $a_n\leq\Tr\tau_n$.

\emph{An average dominating the compression.}
For $g=(g_1,\ldots,g_n)\in SU(2)^n$ and $\pi\in S_n$, set
\begin{equation}
 S_g:=\bigotimes_{j=1}^nU(g_j)^{\otimes4},\qquad
 W_{\pi,g}:=U_\pi S_g
 \in\operatorname{End}(\cK^{\otimes n}).
\label{eq:physical-spin-action}
\end{equation}
Let $\mathbb E$ average over independent normalized Haar variables
$g_j$ and an independent uniform permutation $\pi$.
Haar averaging of the vector action projects onto the fixed vectors, so
$\mathbb E_g S_g=\Pi_{\mathrm{inv}}^{\otimes n}$.  Hence
\[
 \mathbb E W_{\pi,g}
 =\left(\frac1{n!}\sum_\pi U_\pi\right)\Pi_{\mathrm{inv}}^{\otimes n}=P.
\]
The corresponding average on operators is the linear map
\[
 \cE:\operatorname{End}(\cK^{\otimes n})\longrightarrow
       \operatorname{End}(\cK^{\otimes n}),\qquad
 \cE(X):=\mathbb E[W_{\pi,g}XW_{\pi,g}^*].
\]
Using $\mathbb EW_{\pi,g}=\mathbb EW_{\pi,g}^*=P$, expansion gives
\begin{equation}
 \cE(X)-PXP
 =\mathbb E[(W_{\pi,g}-P)X(W_{\pi,g}-P)^*]\geq0,
 \qquad X\geq0.
\label{eq:averaged-compression}
\end{equation}
In particular, $P\sigma_nP\leq\cE(\sigma_n)$.

\emph{Transfer to pair marginals.}
After grouping factors by party, both $S_g$ and $U_\pi$ are tensor
products of unitaries on $A^{\otimes n},B^{\otimes n},C^{\otimes n},
R^{\otimes n}$.  Thus each $W=W_{\pi,g}$ factors across every
pair/complement cut as $W_{XY}\otimes W_{(XY)^c}$, and
\[
 \bigl(\cE(\sigma_n)\bigr)_{(XY)^n}
 =\mathbb E[W_{XY}(\sigma_n)_{(XY)^n}W_{XY}^*].
\]
On the other hand, the zero-extension convention and positivity of
partial trace give
\[
 0\leq\cN_{XY}^{\otimes n}(\tau_n)
 =\Tr_{((XY)^c)^n}(P\sigma_nP)
 \leq\bigl(\cE(\sigma_n)\bigr)_{(XY)^n}.
\]
For positive $X,Y$ on a common Hilbert space,
\begin{equation}
 \norm X_\infty\leq\norm{X+Y}_\infty,
 \qquad \Tr X^2\leq\Tr(X+Y)^2,
\label{eq:positive-norm-monotonicity}
\end{equation}
where the second inequality follows from
$\Tr(X+Y)^2-\Tr X^2=2\Tr XY+\Tr Y^2\geq0$.
More generally, the min--max principle gives
$\lambda_j(X)\leq\lambda_j(X+Y)$ for the decreasingly ordered
eigenvalues, so $\norm X_p\leq\norm{X+Y}_p$ for all
$p\in[1,\infty]$.  Combining this monotonicity with convexity and
unitary invariance gives, for every such $p$,
\[
\begin{aligned}
 \norm{\cN_{XY}^{\otimes n}(\tau_n)}_p
 &\leq\norm{\bigl(\cE(\sigma_n)\bigr)_{(XY)^n}}_p\\
 &\leq\mathbb E\norm{W_{XY}(\sigma_n)_{(XY)^n}W_{XY}^*}_p
 =\norm{(\sigma_n)_{(XY)^n}}_p.
\end{aligned}
\]
Squaring the $p=2$ inequality gives the purity bound.
\end{proof}

Every $W_{\pi,g}$ fixes $\operatorname{Im}J_n$ pointwise, so
$J_n^*\cE(\sigma_n)J_n=\tau_n$: averaging leaves the compressed
operator unchanged.  Averaging alone need not restrict the support to
$\operatorname{Im}P$; for example, it fixes the full-rank maximally
mixed state.  The rest of the proof uses only $\tau_n$ and the conclusions
of \cref{lem:singlet-compression}.

\subsection{Binary types and physical support dimensions}

For $\bm\ell=(\ell_1,\ldots,\ell_n)\in\{0,1\}^n$, set
$\ket{\bm\ell}=f_{\ell_1}\otimes\cdots\otimes f_{\ell_n}
\in\mathcal M^{\otimes n}$.
For integers $0\leq k\leq n$, the normalized \emph{Dicke vector}
$\ket{k}\in\mathcal H_n$ \cite{Dicke} is
\begin{equation}
 \ket{k}=\binom nk^{-1/2}
 \sum_{\substack{\bm\ell\in\{0,1\}^n\\\sum_j\ell_j=k}}\ket{\bm\ell}
 \in\mathcal H_n\subset\mathcal M^{\otimes n}.
\label{eq:dicke-vectors}
\end{equation}
These vectors form an orthonormal basis of $\mathcal H_n$ for the
inherited inner product.
For $\ell\in\{0,1\}$, write
$V_\ell^{AB}=\operatorname{Im}\Pi_\ell^{AB}\subset A\otimes B$ for
the physical spin-$\ell$ subspace.  Define
\begin{equation}
 \mathcal W_{k,n}^{AB}
 :=\bigoplus_{\substack{\bm\ell\in\{0,1\}^n\\\sum_j\ell_j=k}}
       V_{\ell_1}^{AB}\otimes\cdots\otimes V_{\ell_n}^{AB}
 \subset(A\otimes B)^{\otimes n}.
\label{eq:type-output-space}
\end{equation}
Each tensor factor occupies its corresponding physical copy.  Write
$Q_{k,n}^{AB}$ for the orthogonal projection onto
$\mathcal W_{k,n}^{AB}$, namely
\begin{equation}
 Q_{k,n}^{AB}
 :=\sum_{\substack{\bm\ell\in\{0,1\}^n\\\sum_j\ell_j=k}}
       \bigotimes_{j=1}^n\Pi_{\ell_j}^{AB}
 \in\operatorname{End}((A\otimes B)^{\otimes n}).
\label{eq:type-output-projection}
\end{equation}
Define $\mathcal W_{k,n}^{S_a}$ and $Q_{k,n}^{S_a}$ in the same way
for $S_1=AC$ and $S_2=BC$, using the physical projectors
$\Pi_\ell^{S_a}$.  Their dimensions are the same as for $S_0=AB$.

\begin{lemma}[Binary types and their physical marginals]
\label[lemma]{lem:type-marginals}
The spaces $\mathcal W_{k,n}^{AB}$, $0\leq k\leq n$, are mutually
orthogonal, with dimensions
\begin{equation}
 D_k:=\dim\mathcal W_{k,n}^{AB}
 =\rank Q_{k,n}^{AB}=\binom nk3^k.
\label{eq:type-dimensions}
\end{equation}
The marginal of the physical Dicke vector $J_n\ket{k}$ is
\begin{equation}
\begin{aligned}
 \Tr_{(CR)^n}\!\left(J_n\proj{k}J_n^*\right)
 &=\cN_{AB}^{\otimes n}(\proj{k})\\
 &=\frac{Q_{k,n}^{AB}}{D_k}.
\end{aligned}
\label{eq:dicke-marginal}
\end{equation}
Thus its image is $\mathcal W_{k,n}^{AB}$ and its rank is $D_k$.
For an arbitrary operator $\tau_n\in\operatorname{End}(\mathcal H_n)$,
one has
\begin{equation}
 \cN_{AB}^{\otimes n}(\tau_n)
 =\sum_{k=0}^n\langle k,\tau_n k\rangle\,
       \frac{Q_{k,n}^{AB}}{D_k}.
\label{eq:type-output-decomposition}
\end{equation}
In particular, the physical type weights recover the Dicke diagonal:
\[
 \Tr\!\left[Q_{k,n}^{AB}\cN_{AB}^{\otimes n}(\tau_n)\right]
 =\langle k,\tau_n k\rangle.
\]
\end{lemma}

\begin{proof}
The spaces belonging to different spin sequences are orthogonal because
$V_0^{AB}\perp V_1^{AB}$.  A sequence with $k$ entries equal to one
contributes a tensor product of dimension $3^k$, and there are
$\binom nk$ such sequences.  This proves \cref{eq:type-dimensions} and the asserted
orthogonality.

For $\bm\ell,\bm m\in\{0,1\}^n$, tensoring \cref{eq:ab-channel} gives
\[
 \cN_{AB}^{\otimes n}(\ketbra{\bm\ell}{\bm m})
 =\delta_{\bm\ell,\bm m}
   \bigotimes_{j=1}^n\frac{\Pi_{\ell_j}^{AB}}{d_{\ell_j}}.
\]
In the expansion of $\proj{k}$, precisely the $\binom nk$ diagonal
terms survive, each with coefficient $\binom nk^{-1}3^{-k}=D_k^{-1}$.
Their sum is $Q_{k,n}^{AB}/D_k$, proving \cref{eq:dicke-marginal}.
If $k\ne k'$, the spin sequences in $\ket{k}$ and $\ket{k'}$
are disjoint, so
$\cN_{AB}^{\otimes n}(\ketbra{k}{k'})=0$.
Expanding $\tau_n$ in the Dicke basis proves
\cref{eq:type-output-decomposition}.  Finally,
$Q_{k,n}^{AB}Q_{k',n}^{AB}=\delta_{kk'}Q_{k,n}^{AB}$ and
$\Tr Q_{k,n}^{AB}=D_k$ give the type-weight identity.
\end{proof}

Define the \emph{binary entropy} $b:[0,1]\to\mathbb R$ and the function
$f:[0,1]\to\mathbb R$ by
\begin{equation}
 b(u)=-u\log u-(1-u)\log(1-u),\qquad f(u)=b(u)+u\log3,
\label{eq:binary-entropy}
\end{equation}
with $0\log0=0$.  The standard binomial estimate
$\binom nk\leq\e^{nb(k/n)}$ follows by bounding the $k$-th term of a
binomial probability sum by one, including the endpoint cases.  Hence
\begin{equation}
 D_k\leq\e^{nf(k/n)},\qquad f(1/2)=h.
\label{eq:type-dimension-bound}
\end{equation}
For $0<\eta<1/2$, define the \emph{entropy deficit}
\begin{equation}
 \Delta(\eta)=h-f(1/2-\eta)
 =\log2-b(1/2-\eta)+\eta\log3
 =\eta\log3+2\eta^2+O(\eta^4).
\label{eq:deficit}
\end{equation}
In particular $\Delta(\eta)\geq\eta\log3>0$.

\subsection{One-sided frequency caps and the frame bound}

Define the following endomorphisms of $\mathcal H_n$:
\begin{equation}
 P_{\eta,n}=\sum_{k/n\leq1/2-\eta}\proj{k},\qquad
 \Gamma_n=\Gamma^{\otimes n}|_{\mathcal H_n},\qquad
 P_{\eta,n}^{(a)}=\Gamma_n^aP_{\eta,n}\Gamma_n^{-a},
\label{eq:frequency-caps}
\end{equation}
\begin{equation}
 L_{\eta,n}=\sum_{a=0}^2P_{\eta,n}^{(a)}.
\label{eq:frame-operator}
\end{equation}
Here $\Gamma_n$ is unitary, and $P_{\eta,n}$ and $P_{\eta,n}^{(a)}$
are orthogonal projections.  The latter are \emph{spectral caps} for the spin-one
frequency in the three coupling bases.  Their sum
$L_{\eta,n}\in\operatorname{End}(\mathcal H_n)$ is the associated
positive \emph{unit-weight fusion-frame operator}
\cite{CasazzaKutyniokLi}.
To interpret its expectation, put $s_n=\Tr\tau_n$.  If $s_n>0$,
let $\widehat\tau_n=\tau_n/s_n$, and let $K_a\in\{0,\ldots,n\}$
be the outcome of measuring this normalized state in the rotated Dicke
basis $(\Gamma_n^a\ket{k})_{k=0}^n$.  Then
\begin{equation}
 \Tr(\tau_nL_{\eta,n})
 =s_n\sum_{a=0}^2\Pr_{\widehat\tau_n}\{K_a\leq n(1/2-\eta)\}.
\label{eq:frame-probability}
\end{equation}
If $s_n=0$, then $\tau_n=0$.

Since $f$ is increasing on $[0,1/2]$, the
physical marginal of every type in a cap has entropy per copy at most
$h-\Delta(\eta)$.  More precisely, \cref{lem:type-marginals,eq:type-dimension-bound} give
\[
 \frac1n H\!\left(\frac{Q_{k,n}^{AB}}{D_k}\right)
 =\frac{\log D_k}{n}
 \leq f(k/n)\leq h-\Delta(\eta),\qquad k/n\leq1/2-\eta,
\]
and the same bound holds in the other two coupling bases.  We do not need
all types with such low entropy per copy; the one-sided cutoff makes
coherent leakage a binomial tail event.

\begin{lemma}[Marginal bounds imply cap bounds]
\label[lemma]{lem:cap}
Let $0<\eta<1/2$, $\delta>0$, and let $\sigma_n$ be a density
operator on $\cK^{\otimes n}$, with compression
$\tau_n=J_n^*\sigma_nJ_n$.  If its three pair purities are at most
$\e^{-n(h-\delta)}$, then
\begin{equation}
 \Tr(\tau_nL_{\eta,n})
 \leq3\sqrt{n+1}\,\e^{-n(\Delta(\eta)-\delta)/2}.
\label{eq:purity-cap-bound}
\end{equation}
If instead their operator norms are at most $\e^{-n(h-\delta)}$, then
\begin{equation}
 \Tr(\tau_nL_{\eta,n})
 \leq3(n+1)\e^{-n(\Delta(\eta)-\delta)}.
\label{eq:operator-cap-bound}
\end{equation}
\end{lemma}

\begin{proof}
For $AB$, define the orthogonal projection
\[
 R_{\eta,n}^{AB}=\sum_{k/n\leq1/2-\eta}Q_{k,n}^{AB}
 \in\operatorname{End}((A\otimes B)^{\otimes n}).
\]
Its image is the orthogonal sum of the spaces $\mathcal W_{k,n}^{AB}$
with $k/n\leq1/2-\eta$.  By \cref{eq:type-dimension-bound,eq:deficit}, each such type
satisfies
\[
 D_k\leq\e^{nf(k/n)}\leq\e^{nf(1/2-\eta)}
     =\e^{n(h-\Delta(\eta))},
\]
where the second inequality uses
$f'(u)=\log(3(1-u)/u)>0$ for $0<u\leq1/2$, together with
continuity at zero.  Summing the dimensions of at most $n+1$ orthogonal
types and applying \cref{lem:type-marginals} gives
\begin{equation}
\begin{aligned}
 \rank R_{\eta,n}^{AB}
 &=\sum_{k/n\leq1/2-\eta}D_k
 \leq(n+1)\e^{n(h-\Delta(\eta))},\\
 \Tr(\tau_nP_{\eta,n})
 &=\Tr[R_{\eta,n}^{AB}\cN_{AB}^{\otimes n}(\tau_n)].
\end{aligned}
\label{eq:ab-cap-rank-weight}
\end{equation}
The marginal bound from \cref{lem:singlet-compression} and
Hilbert--Schmidt Cauchy--Schwarz yield
\begin{equation}
 \Tr(\tau_nP_{\eta,n})
 \leq\sqrt{\rank R_{\eta,n}^{AB}}\,
       \norm{\cN_{AB}^{\otimes n}(\tau_n)}_2
 \leq\sqrt{n+1}\,\e^{-n(\Delta(\eta)-\delta)/2}.
\label{eq:ab-purity-cap-bound}
\end{equation}
For operator norms, use $\Tr(RX)\leq\rank(R)\norm X_\infty$.
For $S_a\in\{AB,AC,BC\}$, put
$R_{\eta,n}^{S_a}=\sum_{k/n\leq1/2-\eta}Q_{k,n}^{S_a}$.
Tensoring \cref{eq:cyclic-channels} and using the rotated Dicke basis
$(\Gamma_n^a\ket{k})_k$ gives
\begin{equation}
 \Tr(\tau_nP_{\eta,n}^{(a)})
 =\Tr[R_{\eta,n}^{S_a}\cN_{S_a}^{\otimes n}(\tau_n)],
 \qquad a=0,1,2.
\label{eq:rotated-cap-identity}
\end{equation}
The ranks of these three physical cap projections are identical.
The same estimates therefore apply to each pair, and summing proves both
assertions.
\end{proof}

\begin{proposition}[A quadratic frame cost]
\label[proposition]{prop:frame}
For every fixed $0<\eta\leq10^{-4}$, there is $n_0(\eta)$ such that
\begin{equation}
 L_{\eta,n}\geq\e^{-140n\eta^2}I_{\mathcal H_n},
 \qquad n\geq n_0(\eta).
\label{eq:frame-bound}
\end{equation}
\end{proposition}

The proof uses the probability geometry and triangular containment of
\cref{sec:entropy}, followed by the estimates of \cref{sec:toeplitz}.
Each individual
projection has a kernel; the assertion concerns their sum and is uniform
in the vector being tested.  The lower bound is exponentially small, but
its exponent is quadratic in the frequency tolerance.  This is sufficient
because $\Delta(\eta)$ has a positive linear term.

\section{Probability geometry and the two holes}
\label{sec:entropy}

The geometric construction in this section proceeds in two stages.
A point of the Bloch sphere of the singlet multiplicity space $\mathcal M$
determines three nonnegative spin-one probabilities whose sum is $3/2$, so their
triples lie in the triangle $\{u_a\geq0,\ \sum_a u_a=3/2\}$.
The attainable triples form a disk inside this triangle.  A scalar
function on this disk gives the average pair entropy and singles out its
center, whose preimages are two antipodal points of the Bloch sphere.  The frequency
cutoffs used in the frame argument then remove a small triangle about this
center; its preimage consists of two holes in the Bloch sphere of $\mathcal M$.

\subsection{The probability disk}

Recall the singlet multiplicity space $\mathcal M=\mathbb C^2$, the
domain of $J:\mathcal M\to\cK$.  Its \emph{projective line} is
\[
 \CP^1=\{[z]:z\in\mathcal M\setminus\{0\}\},\qquad
 [z]=\{(\lambda z_0,\lambda z_1):\lambda\in\mathbb C\}.
\]
We identify each line with the rank-one orthogonal projector onto it:
\begin{equation}
 [z]\longleftrightarrow\pi_z:=\frac{zz^*}{z^*z},\qquad
 \CP^1\simeq
 \{\pi\in\operatorname{End}(\mathcal M):\pi^2=\pi=\pi^*,\ \Tr\pi=1\}.
\label{eq:projective-projectors}
\end{equation}
This also identifies $\CP^1$ with the \emph{Bloch sphere}.  Indeed, in the
basis $f_0,f_1$, the \emph{Pauli matrices} are
\[
 \sigma_x=\begin{pmatrix}0&1\\1&0\end{pmatrix},\qquad
 \sigma_y=\begin{pmatrix}0&-i\\i&0\end{pmatrix},\qquad
 \sigma_z=\begin{pmatrix}1&0\\0&-1\end{pmatrix}.
\]
Every rank-one orthogonal projector has the unique form
\[
 \pi=\tfrac12(I_{\mathcal M}+r_x\sigma_x+r_y\sigma_y+r_z\sigma_z),
 \qquad r=(r_x,r_y,r_z)\in\mathbb R^3,\quad |r|=1.
\]
The projectors $\pi_v$ and $\pi_{\bar v}$ are antipodal points of
this sphere.  Their physical counterparts are $J\pi_vJ^*$ and
$J\pi_{\bar v}J^*$.

For $a=0,1,2$, define the \emph{spin distribution} in the $a$-th
coupling basis by
\begin{equation}
\begin{aligned}
 q^{(a)}&:\CP^1\longrightarrow
 \{(p_0,p_1)\in[0,1]^2:p_0+p_1=1\},\\
 q^{(a)}([z])_\ell
 &=\Tr\!\left(\pi_z\Gamma^a\proj{f_\ell}\Gamma^{-a}\right)
 =\frac{|(\Gamma^{-a}z)_\ell|^2}{z^*z},\qquad \ell=0,1.
\end{aligned}
\label{eq:spin-distributions}
\end{equation}
For a unit representative $z$, its two components give the spin-zero
and spin-one probabilities for the $S_a$ pair marginal of $\proj{Jz}$,
where $S_0=AB,S_1=AC,S_2=BC$ as in \cref{lem:cyclic-channels}.
Its \emph{spin-one probability} is the function
\begin{equation}
 u_a:\CP^1\longrightarrow[0,1],\qquad
 u_a([z]):=q^{(a)}([z])_1,\qquad a=0,1,2.
\label{eq:spin-one-probability}
\end{equation}
Both definitions are independent of the nonzero representative $z$.
We abbreviate their values at $[z]$, or equivalently at $\pi_z$,
as $q^{(a)}(z)$ and $u_a(z)$.  The physical pair entropy is
$f(u_a(z))$.

Write $\mathcal F([z])=(u_0(z),u_1(z),u_2(z))$ and
$c=(1/2,1/2,1/2)$.  Substituting the Bloch coordinates into
\cref{eq:spin-distributions,eq:spin-one-probability} gives
\begin{equation}
 u_0=\frac12-\frac{r_z}{2},\qquad
 u_1=\frac12+\frac{r_z-\sqrt3r_x}{4},\qquad
 u_2=\frac12+\frac{r_z+\sqrt3r_x}{4}.
\label{eq:bloch-probabilities}
\end{equation}
Thus $\mathcal F$ forgets $r_y$, and
\begin{equation}
 \sum_{a=0}^2u_a=\frac32,\qquad
 \norm{\mathcal F([z])-c}^2=\frac38(r_x^2+r_z^2).
\label{eq:probability-radius}
\end{equation}
Projection of the sphere onto the $(r_x,r_z)$ plane fills the unit
disk, and the remaining affine map scales all distances by
$\sqrt{3/8}$.  Consequently the exact image of the Bloch sphere under
$\mathcal F$ is the filled \emph{probability disk}
\begin{equation}
 D_{\mathrm{prob}}:=\mathcal F(\CP^1)
 =\left\{u\in\mathbb R^3:\sum_{a=0}^2u_a=\frac32,
                   \ \norm{u-c}^2\leq\frac38\right\}.
\label{eq:probability-disk}
\end{equation}
Every interior point has two preimages, distinguished by the sign of
$r_y$; every boundary point has one.  The two preimages are related by
complex conjugation in the basis $f_0,f_1$.  In particular, the center
$c$ lifts to $[v]$ and $[\bar v]$, the poles $r_y=-1$ and
$r_y=1$, whose physical representatives are $\Omega$ and
$\bar\Omega$.

The sum constraint and nonnegativity alone describe the larger triangle
$\{u_a\geq0,\ \sum_a u_a=3/2\}$.  Adding $u_a\leq1$ cuts off
its corners and leaves a hexagon.  The disk inside this hexagon is the
additional compatibility restriction imposed by realizing all three
probabilities at the same point of the Bloch sphere of $\mathcal M$; see
\cref{fig:probability-disk}.  For example, $z=f_0$, whose physical
state is $e_0$, gives the boundary point $(0,3/4,3/4)$.
The triple $(1,0,1/2)$ lies in the hexagon but cannot occur, since its
squared distance from $c$ is $1/2>3/8$.

We normalize the \emph{Fubini--Study distance}
$d_{FS}:\CP^1\times\CP^1\to[0,\pi/2]$ by
$d_{FS}([z],[w])=\arccos|\langle z,w\rangle|$ for unit vectors.

\subsection{Entropy certainty as geometric motivation}

This subsection explains the distinguished role of the two target
directions.  Its entropy inequality is not needed for the frame estimate,
which uses the individual frequency constraints and triangular containment
below.

The second map of this construction assigns the \emph{average pair entropy} to each attainable
probability triple:
\begin{equation}
 \mathscr H:D_{\mathrm{prob}}\longrightarrow\mathbb R,\qquad
 \mathscr H(u)=\frac13\sum_{a=0}^2 f(u_a)
             =\frac12\log3+\frac13\sum_{a=0}^2 b(u_a).
\label{eq:average-entropy}
\end{equation}
The last identity uses $\sum_a u_a=3/2$.  For the physical pure state
$\proj{Jz}$, the composition $\mathscr H(\mathcal F([z]))$ is the average
of the entropies of its $AB,AC,BC$ marginals.  Equivalently, it is the
mean entanglement entropy across the three two-versus-two cuts.
Each binary entropy in \cref{eq:average-entropy} is at most $\log2$,
so the unique maximizing triple is $c$, with value
$h=\log2+\tfrac12\log3$.  Its two lifts are precisely the maximizing
directions within the pure singlet family.  The following proposition
also quantifies the loss away from them.

\begin{proposition}[Entropy certainty and its equality set]
\label[proposition]{prop:entropy}
For every unit $z\in\mathcal M$,
\begin{equation}
 \frac13\sum_{a=0}^2u_a(z)=\frac12,\qquad
 \frac13\sum_{a=0}^2f(u_a(z))\leq h.
\label{eq:entropy-certainty}
\end{equation}
Equality in the entropy inequality holds exactly at $[v]$ and
$[\bar v]$.  More precisely, if $z=\alpha v+\beta\bar v$ with
$\alpha,\beta\in\mathbb C$ and
$|\alpha|^2+|\beta|^2=1$, then
\begin{equation}
 h-\frac13\sum_{a=0}^2f(u_a(z))\geq|\alpha\beta|^2.
\label{eq:entropy-stability}
\end{equation}
\end{proposition}

\begin{proof}
The probability identity is \cref{eq:probability-radius}, and the entropy
bound follows from \cref{eq:average-entropy}.  For the quantitative
version, $b'(1/2)=0$ and $b''(u)=-1/[u(1-u)]\leq-4$ give
\[
 b(1/2+d)\leq\log2-2d^2,\qquad -1/2\leq d\leq1/2,
\]
including the endpoints by continuity.  Averaging and using
\cref{eq:probability-radius} yields the geometric form of the bound:
\begin{equation}
 h-\mathscr H(\mathcal F([z]))\geq\frac23\norm{\mathcal F([z])-c}^2
     =\frac14(r_x^2+r_z^2)=|\alpha\beta|^2.
\label{eq:entropy-radial-bound}
\end{equation}
For the last equality, $r_y=|\beta|^2-|\alpha|^2$ in the orthonormal
basis $v,\bar v$, so $r_x^2+r_z^2=1-r_y^2=4|\alpha\beta|^2$.
Equality in the entropy inequality forces $\mathcal F([z])=c$; its two
preimages both attain equality.
\end{proof}

In particular, if $\vartheta=d_{FS}([z],\{[v],[\bar v]\})$, then
$0\leq\vartheta\leq\pi/4$ and
$|\alpha\beta|^2=\tfrac14\sin^2(2\vartheta)\geq4\vartheta^2/\pi^2$.
This supplies global quadratic stability.  Near the center of the
probability disk, the more precise expansion is
\begin{equation}
 h-\mathscr H(c+d)=\frac23\norm d^2+O(\norm d^4),
 \qquad \sum_{a=0}^2d_a=0.
\label{eq:entropy-contours}
\end{equation}
The linear terms from $f$ cancel, and the binary entropy is even about
$1/2$, so there is no cubic term.  Thus nearby entropy contours are
approximately circles.  The triangular geometry needed below instead
comes from imposing three separate frequency inequalities.

\subsection{Exact coordinates and triangular containment}

For $0<t\leq10^{-3}$, put
\begin{equation}
 K_t=\{[z]\in\CP^1:\min_a u_a(z)\leq1/2-t\},\qquad
 \cD_t=\CP^1\setminus K_t.
\label{eq:Kt}
\end{equation}
Here $t$ is a \emph{frequency gap}, as is $\eta$ in \cref{eq:frequency-caps}; the
corresponding entropy gap is $\Delta(t)$.
In the probability plane, these three strict lower bounds define the
open equilateral \emph{frequency triangle}
\begin{equation}
 \mathscr T_t=\left\{c+d:\sum_{a=0}^2d_a=0,\ d_a>-t
                                      \text{ for }a=0,1,2\right\},
 \qquad \cD_t=\mathcal F^{-1}(\mathscr T_t).
\label{eq:probability-triangle}
\end{equation}
Its closure has vertices $c+(2t,-t,-t)$ and their cyclic permutations,
and circumradius $\sqrt6\,t$.  Since $D_{\mathrm{prob}}$ has radius
$\sqrt{3/8}=\sqrt6/4$, the closed triangle lies strictly inside the
disk whenever $0<t<1/4$, in particular throughout our range of $t$.
The map $\mathcal F$ restricts to a diffeomorphism from each open hemisphere
$r_y>0$ or $r_y<0$ onto the interior of the disk.  Hence
$\cD_t$ has exactly two components in the Bloch sphere of $\mathcal M$:
the two \emph{holes}, one about each of $[v]$ and $[\bar v]$.

\begin{figure}[htbp]
\centering
\begin{tikzpicture}[x=1.65cm,y=1.65cm,font=\small]
 \filldraw[fill=qtgray!4,draw=qtgray,thick]
   (0,2)--({-sqrt(3)},-1)--({sqrt(3)},-1)--cycle;
 \filldraw[fill=qtgray!10,draw=qtgray,dashed,thick]
   ({1/sqrt(3)},1)--({-1/sqrt(3)},1)--({-2/sqrt(3)},0)
   --({-1/sqrt(3)},-1)--({1/sqrt(3)},-1)--({2/sqrt(3)},0)--cycle;
 \filldraw[fill=qtblue!12,draw=qtblue,thick] (0,0) circle (1);
 \filldraw[fill=qtorange!25,draw=qtorange!90!black,thick]
   (0,0.56)--({-0.28*sqrt(3)},-0.28)--({0.28*sqrt(3)},-0.28)--cycle;
 \fill (0,0) circle (1.2pt);
 \node[below=2pt] at (0,0) {$c$};
 \node[above=3pt] at (0,2) {$(\tfrac32,0,0)$};
 \node[below=3pt] at ({-sqrt(3)},-1) {$(0,\tfrac32,0)$};
 \node[below=3pt] at ({sqrt(3)},-1) {$(0,0,\tfrac32)$};
 \fill[qtblue] (0,-1) circle (1.5pt);
 \node[below=3pt,align=center] at (0,-1)
   {$(0,\tfrac34,\tfrac34)$\\attainable};
 \draw[qtred,very thick]
   ({1/sqrt(3)-0.035},0.965)--({1/sqrt(3)+0.035},1.035)
   ({1/sqrt(3)-0.035},1.035)--({1/sqrt(3)+0.035},0.965);
 \draw[qtred,thin] ({1/sqrt(3)+0.05},1.04)--(1.05,1.38);
 \node[anchor=west,qtred,align=left] at (1.04,1.4)
   {$(1,0,\tfrac12)$\\unattainable};
 \begin{scope}[xshift=4.6cm]
  \draw[qtgray,thick] (0,0.85)--(0.3,0.85);
  \node[anchor=west,align=left] at (0.4,0.85)
    {Nonnegative triangle\\$u_a\geq0,\ \sum_a u_a=3/2$};
  \draw[qtgray,dashed,thick] (0,0.15)--(0.3,0.15);
  \node[anchor=west,align=left] at (0.4,0.15)
    {Hexagon\\also $u_a\leq1$};
  \draw[qtblue,thick] (0,-0.55)--(0.3,-0.55);
  \node[anchor=west,align=left] at (0.4,-0.55)
    {Probability disk $D_{\mathrm{prob}}$\\compatible triples};
  \draw[qtorange!90!black,thick] (0,-1.25)--(0.3,-1.25);
  \node[anchor=west,align=left] at (0.4,-1.25)
    {Frequency triangle $\mathscr T_t$\\$u_a>1/2-t$ for every $a$};
 \end{scope}
\end{tikzpicture}
\caption{The probability plane, up to a uniform change of scale.
The disk is the exact image of the Bloch sphere of $\mathcal M$.  Its center has the two
preimages $[v]$ and $[\bar v]$, and the open orange triangle lifts
to the two holes on that sphere.  The triangle is drawn with $t=0.14$ for visibility;
the estimates use $t\leq10^{-3}$.}
\label{fig:probability-disk}
\end{figure}
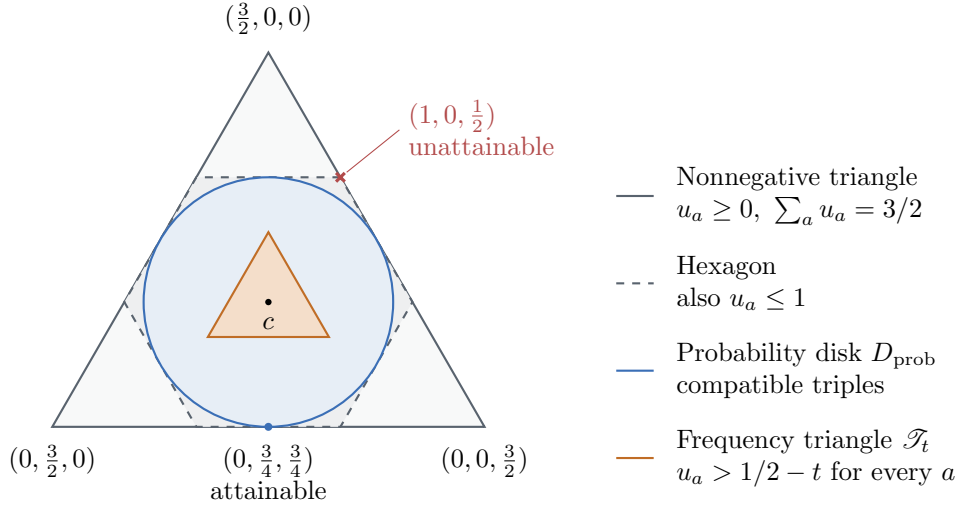

The probability triangle $\mathscr T_t$ is a small equilateral triangle,
but the two components of its preimage $\cD_t$ have curved sides
in affine coordinates on $\CP^1$.  We now enclose these components in slightly
larger affine triangles, whose straight sides will be used in the
polynomial restriction argument.
Use the \emph{affine parametrizations} $\chi_\pm:\mathbb C\to\CP^1$ given by
$\chi_\pm(x)=[z_\pm(x)]$, with unit-vector representatives
$z_\pm(x)\in\mathcal M$ defined for $x\in\mathbb C$ by
\begin{equation}
 z_+(x)=\frac{v-x\bar v}{\sqrt{1+|x|^2}},\qquad
 z_-(x)=\frac{\bar v-xv}{\sqrt{1+|x|^2}}.
\label{eq:charts}
\end{equation}
The maps $\chi_+$ and $\chi_-$ parametrize the complements of
$[\bar v]$ and $[v]$, respectively.  The images of the two closed
coordinate disks $|x|\leq1$ cover $\CP^1$.
Using \cref{eq:cyclic-eigenvalues}, one obtains the exact identities
\begin{equation}
 u_a(z_+(x))=\frac12+
       \frac{\operatorname{Re}(\omega^a x)}{1+|x|^2},\qquad
 u_a(z_-(x))=\frac12+
       \frac{\operatorname{Re}(\omega^{-a}x)}{1+|x|^2}.
\label{eq:trine}
\end{equation}
Thus the two local pictures differ only by cyclic orientation.
For $\tau>0$, define the real-affine functions
$\ell_a:\mathbb C\to\mathbb R$ and the open equilateral triangle
$\mathcal T_\tau\subset\mathbb C$ of inradius $\tau$ by
\begin{equation}
 \ell_a(x)=\tau+\operatorname{Re}(\omega^a x),\qquad
 \mathcal T_\tau=\{x\in\mathbb C:\ell_a(x)>0\text{ for }a=0,1,2\}.
\label{eq:triangle}
\end{equation}

\begin{lemma}[Two small triangular holes]
\label[lemma]{lem:shape}
Every point of $\cD_t$, in one of the charts $\chi_\pm$ with
$|x|\leq1$, satisfies
\begin{equation}
 |x|<4t,\qquad
 \operatorname{Re}(\omega^a x)>-t(1+16t^2),\quad a=0,1,2.
\label{eq:containment}
\end{equation}
In particular, the holes lie in the two disjoint chart images
$\chi_+(\mathcal T_{1.001t})$ and $\chi_-(\mathcal T_{1.001t})$
about $[v]$ and $[\bar v]$.
\end{lemma}

\begin{proof}
In the $\chi_-$ chart, replacing $a$ by $-a$ permutes
$\{0,1,2\}$, so the three inequalities in \cref{eq:trine} define
the same set as those in \cref{eq:containment}.
For any complex $x$, one of the three numbers
$\operatorname{Re}(\omega^a x)$ is at most $-|x|/2$.
Membership in $\cD_t$ and \cref{eq:trine} imply
\[
 \frac{|x|}{2(1+|x|^2)}<t.
\]
Since $|x|\leq1$, this gives $|x|<4t$.  Substitute back into the
three defining inequalities to obtain \cref{eq:containment}.
Finally $1+16t^2\leq1.000016<1.001$.  Each containing triangle has
circumradius $2.002t<1$; their chart images are disjoint, since coordinates
of a point in both charts have reciprocal absolute values.
\end{proof}

The containment and the enlarged triangles used below are
illustrated together in \cref{fig:triangle}.  The triangles
$\mathcal T_\tau\subset\mathbb C$ are containing domains in an affine
chart, whereas $\mathscr T_t$ in \cref{eq:probability-triangle} is the
exact frequency region in the probability plane.

The containment also separates the roles of entropy and frequency in the
proof.
The holes are defined by $\min_a u_a>1/2-t$, and their boundaries are
not level sets of the average entropy: by \cref{eq:entropy-contours},
the average deficit is $t^2+O(t^4)$ at a side midpoint of
$\mathscr T_t$, but $4t^2+O(t^4)$ at a vertex.
On $K_t$, the complement of the two holes, some $u_a\leq1/2-t$, and
monotonicity of $f$ on $[0,1/2]$ gives the individual pair bound
\[
 f(u_a)\leq h-\Delta(t),\qquad
 \Delta(t)=t\log3+2t^2+O(t^4).
\]
These one-sided frequency conditions match the caps in \cref{eq:frequency-caps}.
The next section shows that coherent-state mass cannot concentrate
entirely in the two holes, at a frame cost quadratic in the frequency
tolerance.  It is the linear term of the individual entropy gap that
beats this quadratic cost in \cref{sec:contradiction}.

\section{The coherent-state frame estimate}
\label{sec:toeplitz}

\subsection{Coherent states and polynomial restriction}

Recall $\mathcal H_n=\Sym^n(\mathcal M)$ and
$N_n=\dim\mathcal H_n=n+1$ from \cref{sec:reduction}.
For a unit vector $z\in\mathcal M$, define its \emph{coherent state}
$e_z^{(n)}$ and, for $\psi\in\mathcal H_n$, the \emph{squared coherent overlap}
$W_{\psi,n}$ by
\begin{equation}
 e_z^{(n)}=z^{\otimes n}\in\mathcal H_n,\qquad
 W_{\psi,n}([z])=|\langle e_z^{(n)},\psi\rangle|^2.
\label{eq:coherent-overlaps}
\end{equation}
Although $e_z^{(n)}$ changes by a phase when the representative changes,
its rank-one projector is
$\proj{e_z^{(n)}}=\pi_z^{\otimes n}|_{\mathcal H_n}$ and depends
only on $[z]$.  Likewise,
$W_{\psi,n}:\CP^1\to[0,\infty)$ is well defined.
In either affine chart $\chi_\pm$ from \cref{eq:charts}, write
$W_{\psi,n}(x)=W_{\psi,n}(\chi_\pm(x))$.  Then
\begin{equation}
 W_{\psi,n}(x)=\frac{|P_\psi(x)|^2}{(1+|x|^2)^n},
\label{eq:weighted-poly}
\end{equation}
where $P_\psi\in\mathbb C[x]$ is a polynomial in one complex variable,
of degree at most $n$, whose coefficients depend on the chart.
More explicitly, the two choices are
\[
 P_{\psi,+}(x)=\langle\psi,(v-x\bar v)^{\otimes n}\rangle,\qquad
 P_{\psi,-}(x)=\langle\psi,(\bar v-xv)^{\otimes n}\rangle.
\]
These are holomorphic in $x$, since our inner product is linear in its
second argument, and their squared absolute values give
\cref{eq:weighted-poly}.  Before choosing an affine chart, the overlaps
are represented by \emph{homogeneous binary polynomials} of degree $n$,
meaning homogeneous polynomials in two complex coordinates.

Let $\mu_{FS}$ be the \emph{Fubini--Study probability measure} on $\CP^1$.
Its chart density is
$\pi^{-1}(1+|x|^2)^{-2}$ with respect to planar area.  The auxiliary
$SU(2)$ action on the multiplicity space $\mathcal M$ induces the
unitary representation $g\mapsto g^{\otimes n}|_{\mathcal H_n}$ on
$\mathcal H_n$, which is irreducible of spin $n/2$.  This auxiliary action
is distinct from the physical diagonal spin action of
\cref{eq:physical-spin-action}, which is
trivial on the singlets.  Invariance of $\mu_{FS}$, Schur's lemma and
the trace give the \emph{coherent-state resolution of the identity}:
\begin{equation}
 N_n\int_{\CP^1}\proj{e_z^{(n)}}\,\dd\mu_{FS}(z)=I_{\mathcal H_n}.
\label{eq:coherent-resolution}
\end{equation}
Equivalently, this identity follows by integrating the Dicke monomials
using the binomial theorem.

\Needspace{7\baselineskip}
For a measurable set $K\subset\CP^1$,
write $\1_K:\CP^1\to\{0,1\}$ for its \emph{indicator} and define the positive
\emph{Toeplitz operator}
\begin{equation}
 T_n(\1_K)=N_n\int_K\proj{e_z^{(n)}}\,\dd\mu_{FS}(z)
 \in\operatorname{End}(\mathcal H_n).
\label{eq:toeplitz-definition}
\end{equation}
We will first bound this positive operator from below, then compare it to
$L_{\eta,n}$.

\Needspace{12\baselineskip}

\begin{theorem}[Bernstein--Markov restriction on the projective line]
\label[theorem]{thm:BM}
Let $K\subset\CP^1$ be the closure of a nonempty open set with
Lipschitz boundary.  For each $\zeta>0$, there is $C_{K,\zeta}<\infty$ such that
\begin{equation}
 \sup_KW_{\psi,n}\leq C_{K,\zeta}\e^{n\zeta}
                    \int_KW_{\psi,n}\,\dd\mu_{FS}
\label{eq:bernstein-markov}
\end{equation}
for every $n\geq1$ and $\psi\in\mathcal H_n$.
\end{theorem}

This is an elementary projective-line instance of weighted
Bernstein--Markov theory.  For the general criterion, see
\cite[Sec.~1.2, Theorem~1.14]{BermanBoucksomWittNystrom};
mass-density conditions are discussed in
\cite[Sec.~3]{BloomLevenbergPiazzonWielonsky}.
We include a proof that keeps the polynomial factor explicit before it is
absorbed into the arbitrarily small exponential loss.

\begin{proof}
\emph{Evaluation on a real square.}
For integers $k\geq0$, the \emph{normalized shifted Legendre polynomials}
$L_k\in\mathbb R[u]$, defined by
\[
 L_k(u)=\frac{\sqrt{2k+1}}{k!}
       \frac{\mathrm d^k}{\mathrm du^k}[u^k(u-1)^k],
\]
are orthonormal in $L^2([0,1])$, as integration by parts verifies.
Their defining formula gives $L_k(0)=(-1)^k\sqrt{2k+1}$.
Tensor expansion and Cauchy--Schwarz give, for every complex
polynomial $q\in\mathbb C[u_1,u_2]$ of degree at most $n$ in each variable,
\begin{equation}
 |q(0,0)|^2\leq
 \left(\sum_{k=0}^n(2k+1)\right)^2
       \int_{[0,1]^2}|q(u)|^2\,\dd u
 =(n+1)^4\int_{[0,1]^2}|q(u)|^2\,\dd u.
\label{eq:square-evaluation}
\end{equation}

\emph{Small inward parallelograms.}
Fix $a>0$.  Choose finitely many projective affine charts obtained
from orthonormal bases of $\mathcal M$, and smaller relatively compact
neighborhoods in them covering $K$.  In each chart, suppressing its
index, write
\[
 W(x)=|P(x)|^2\e^{-n\varphi(x)},\qquad
 \varphi(x)=\log(1+|x|^2),\qquad
 \dd\mu_{FS}(x)=\rho(x)\,\dd A(x),
\]
where $\dd A$ is planar area.  On these finitely many neighborhoods
and slightly larger compact neighborhoods, the densities satisfy a
common bound $\rho\geq d_K>0$.

The Lipschitz boundary gives a finite collection of boundary
patches in these charts.  After a translation and an orthogonal change
of real coordinates $(s,r)$, each patch has the form
\[
 K=\{(s,r):r\geq g(s)\},\qquad |g(s)-g(s')|\leq L|s-s'|,
\]
locally, where one finite $L\geq0$ works for all the patches.
Choose smaller patches whose closures lie inside the original patches
and which still cover $\partial K$.  Put $m=L+1$.  For a point
$x=(s_0,r_0)\in K$ in a smaller boundary patch, define
\[
 Q_x=\left\{\bigl(s_0+\delta(u_1-u_2),\,
                   r_0+m\delta(u_1+u_2)\bigr):
                   0\leq u_1,u_2\leq1\right\}.
\]
A single sufficiently small $\delta>0$ keeps every such parallelogram
inside its original patch, by the positive margins of the finite cover.
Moreover, every point $(s,r)\in Q_x$ satisfies
\[
 g(s)\leq g(s_0)+L\delta|u_1-u_2|
       \leq r_0+L\delta(u_1+u_2)
       \leq r,
\]
so $Q_x\subset K$.  Its two edge vectors are independent, with
absolute determinant $2m\delta^2$.  The portion of $K$ outside
the smaller boundary patches is a compact subset of the interior.
There we use squares $x+\delta[0,1]^2$ in the chosen chart, reducing
$\delta$ uniformly if necessary to keep them in $K$ and in that
chart.

Returning to the original chart coordinates, both constructions have
the form $Q_x=x+T_x[0,1]^2$, where
$T_x\in\operatorname{GL}_2(\mathbb R)$.  Translations and orthogonal
changes of coordinates preserve areas and distances.  Uniform continuity
of the finitely many chart potentials allows one final reduction of
$\delta$, depending on $a$, such that
\begin{equation}
 Q_x\subset K,\qquad
 |\det_{\mathbb R}T_x|\geq\delta^2=:c_{K,a}>0,\qquad
 \sup_{y\in Q_x}|\varphi(y)-\varphi(x)|\leq a.
\label{eq:inward-parallelogram}
\end{equation}
All these choices depend only on $K$, the fixed charts, and $a$,
and are independent of $n$, $\psi$, and $x$.

\emph{Restoring the weight.}
Writing the one complex chart coordinate as two real coordinates, the
real-affine substitution defines a polynomial
$q_x(u)=P(x+T_xu)$ of total degree at most $n$, hence of degree
at most $n$ in each real variable.  Since $q_x(0)=P(x)$,
\cref{eq:square-evaluation} and the change of variables $y=x+T_xu$ give
\[
 |P(x)|^2\leq
 \frac{(n+1)^4}{|\det_{\mathbb R}T_x|}
 \int_{Q_x}|P(y)|^2\,\dd A(y).
\]
Multiplying by $\e^{-n\varphi(x)}$, and using
$\varphi(y)-\varphi(x)\leq a$ and $\rho(y)\geq d_K$, yields
\[
\begin{aligned}
 W(x)&\leq\frac{(n+1)^4}{c_{K,a}}
  \int_{Q_x}W(y)\e^{n(\varphi(y)-\varphi(x))}\,\dd A(y)\\
 &\leq\frac{(n+1)^4\e^{na}}{c_{K,a}d_K}
  \int_KW\,\dd\mu_{FS}.
\end{aligned}
\]
Set $a=\zeta/2$, take the supremum over $x$, and use
$\sup_{n\geq1}(n+1)^4\e^{-n\zeta/2}<\infty$ to obtain
\cref{eq:bernstein-markov}.
\end{proof}

The constant may depend on the domain and on $\zeta$.  We always fix
both before letting $n\to\infty$.

\subsection{Triangular holes and the coherent-state lower bound}

For fixed $0<t\leq10^{-3}$, set
\begin{equation}
 \tau=1.001t,\qquad
 \widetilde K_t=\CP^1\setminus
       \bigl(\chi_+(\mathcal T_\tau)\cup\chi_-(\mathcal T_\tau)\bigr).
\label{eq:enlarged-holes}
\end{equation}
The two chart triangles have disjoint closures, since their closures
lie in $|x|<1$.  By \cref{lem:shape}, they contain the actual holes,
so $\widetilde K_t\subset K_t$: every point of $\widetilde K_t$
has $u_a\leq1/2-t$ in at least one basis.  The image of each triangle
boundary belongs to $\widetilde K_t$.

This complement is the closure of a nonempty open set with Lipschitz
boundary.  At a triangle vertex, suitable orthogonal real coordinates
express the exterior as $r\geq-\sqrt3|s|$ locally; at an edge the
boundary is a straight line.  Thus \cref{thm:BM} applies directly,
without smoothing the triangle corners.  The set is fixed before
letting $n\to\infty$.

\begin{figure}[htbp]
\centering
\begin{tikzpicture}[font=\small]
 \begin{scope}
  \node[font=\small\bfseries] at (0.15,2.9)
    {(a) Enlarged hole near $[v]$};
  \fill[qtblue!7] (-3.0,-2.45) rectangle (3.3,2.45);
  \filldraw[fill=white,draw=qtblue,thick,dashed]
    (2.24,0)--(-1.12,1.9399)--(-1.12,-1.9399)--cycle;
  \filldraw[fill=qtred!12,draw=qtred,thick]
    (1.25,0) .. controls (0.575,0.2742) and (-0.05,0.6351) .. (-0.625,1.0825)
    .. controls (-0.525,0.3608) and (-0.525,-0.3608) .. (-0.625,-1.0825)
    .. controls (-0.05,-0.6351) and (0.575,-0.2742) .. (1.25,0);
  \fill (0,0) circle (1.2pt);
  \node at (-2.0,-1.4) {$0\mapsto[v]$};
  \draw[qtgray,-{Latex[length=1.5mm]}] (-1.5,-1.12)--(-0.08,-0.06);
  \node[qtblue!70!black] at (2.35,1.8) {$\widetilde K_t$};
  \draw[qtred,thick] (-2.4,-3.0)--(-1.9,-3.0);
  \node[anchor=west] at (-1.8,-3.0) {actual hole};
  \draw[qtblue,thick,dashed] (0.7,-3.0)--(1.2,-3.0);
  \node[anchor=west] at (1.3,-3.0) {$\mathcal T_\tau$};
 \end{scope}
 \begin{scope}[xshift=7.1cm]
  \node[font=\small\bfseries] at (0.15,2.9)
    {(b) Partition and leakage};
  \begin{scope}
   \clip (-3.0,-2.45) rectangle (3.3,2.45);
   \fill[qtblue!12] (0,0)--(-10,17.32)--(-10,-17.32)--cycle;
   \fill[qtorange!14] (0,0)--(20,0)--(-10,17.32)--cycle;
   \fill[qtgreen!12] (0,0)--(-10,-17.32)--(20,0)--cycle;
   \draw[qtgray,dashed] (0,0)--(4,0);
   \draw[qtgray,dashed] (0,0)--(-2,3.4641);
   \draw[qtgray,dashed] (0,0)--(-2,-3.4641);
  \end{scope}
  \filldraw[fill=white,draw=qtblue,thick,dashed]
    (2.24,0)--(-1.12,1.9399)--(-1.12,-1.9399)--cycle;
  \node[align=center] at (0,0) {removed\\$\mathcal T_\tau$};
  \node[align=center] at (-2.05,0) {$E_0$\\$u_0$ smallest};
  \node[align=center] at (1.6,1.55) {$E_1$\\$u_1$ smallest};
  \node[align=center] at (1.6,-1.55) {$E_2$\\$u_2$ smallest};
  \fill[qtred!18] (1.25,-3.4) rectangle (2.85,-3.15);
  \draw[qtarrow] (-2.4,-3.4)--(3.15,-3.4) node[right] {$k/n$};
  \foreach \x in {-1.9,-0.4,1.25}
    \draw (\x,-3.47)--(\x,-3.33);
  \node[below] at (-1.9,-3.47) {$q=u_a(z)$};
  \node[below] at (-0.4,-3.47) {$\frac12-t$};
  \node[below] at (1.25,-3.47) {$\frac12-\eta$};
  \draw[decorate,decoration={brace,amplitude=3pt}]
    (-0.4,-3.08)--(1.25,-3.08) node[midway,above=4pt] {$t-\eta$};
  \node[qtred] at (2.25,-2.95) {leakage};
 \end{scope}
\end{tikzpicture}
\caption{One local chart, with curvature and gaps exaggerated.
(a) The actual hole has curved sides and is contained in
the open triangle $\mathcal T_\tau$, where $\tau=1.001t$; the shaded
exterior, including the triangle boundary, is $\widetilde K_t$ locally.
(b) The regions $E_a$ used in \cref{subsec:leakage} select the smallest
$u_a$.  The inset shows the leakage event $k/n>\frac12-\eta$,
separated from its mean $q=u_a(z)\leq\frac12-t$ by at least
$t-\eta$.  The $\chi_-$ chart near $[\bar v]$ interchanges labels
$1$ and $2$.}
\label{fig:triangle}
\end{figure}
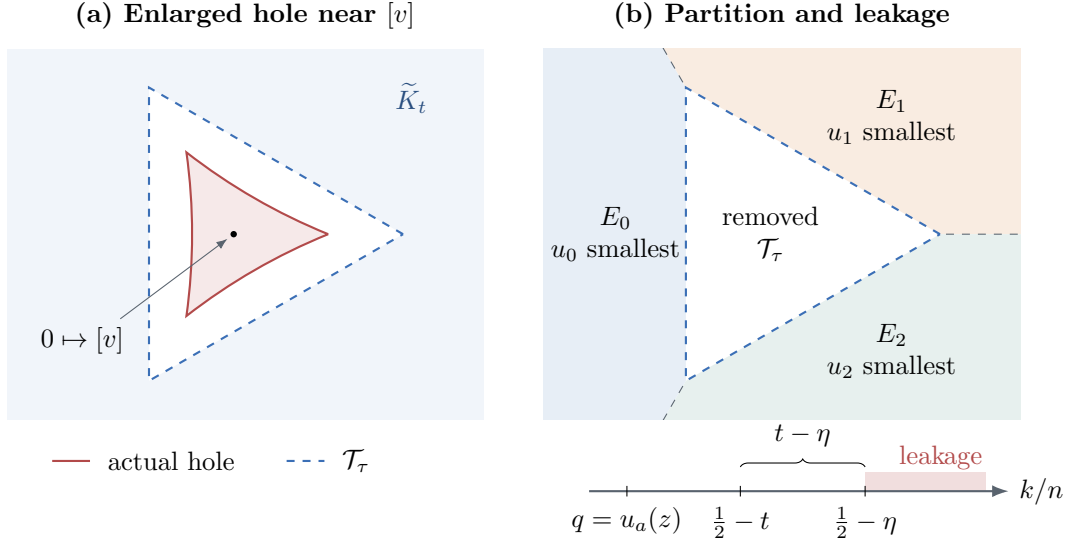

\begin{lemma}[A triangle barrier]
\label[lemma]{lem:polynomial-growth}
For every $0<t\leq10^{-3}$, every $n\geq1$, and every
$\psi\in\mathcal H_n$,
\begin{equation}
 \sup_{\CP^1}W_{\psi,n}
 \leq\exp\!\left(\frac43n(1.001t)^2\right)
      \sup_{\widetilde K_t}W_{\psi,n}.
\label{eq:triangle-growth}
\end{equation}
\end{lemma}

\begin{proof}
For nonzero $\psi$, let $M=\sup_{\widetilde K_t}W_{\psi,n}>0$.
In either affine chart, define
$u:\mathbb C\to\mathbb R\cup\{-\infty\}$ by
\begin{equation}
 u(x)=\frac1n\log\frac{W_{\psi,n}(x)}M
 =\frac1n\log|P_\psi(x)|^2-\log(1+|x|^2)-\frac1n\log M.
\label{eq:logarithmic-growth}
\end{equation}
Allow $u=-\infty$ at zeros of the polynomial.  Since
$\log|P_\psi|^2$ is subharmonic and
$\Delta_x\log(1+|x|^2)=4/(1+|x|^2)^2\leq4$,
$\Delta_x u\geq-4$ distributionally.

On the closed triangle from \cref{eq:triangle}, define the \emph{barrier}
$\mathcal U_\tau:\overline{\mathcal T_\tau}\to[0,\infty)$ by
\begin{equation}
 \mathcal U_\tau(x)=\frac{\ell_0(x)\ell_1(x)\ell_2(x)}{3\tau}.
\label{eq:triangle-torsion}
\end{equation}
Writing $x=X+iY$ gives
\[
 \ell_0\ell_1\ell_2
 =\tau^3-\frac{3\tau}{4}(X^2+Y^2)+\frac14\operatorname{Re}(x^3).
\]
Consequently $-\Delta_x\mathcal U_\tau=1$, with zero boundary values.
Thus $\mathcal U_\tau$ is the classical Dirichlet torsion function of
the equilateral triangle; see \cite{VanDenBergBucur} for the general
torsion-function setting.  Its explicit form makes the constant below
elementary.
Since $\ell_0+\ell_1+\ell_2=3\tau$, the arithmetic--geometric mean
inequality shows
\begin{equation}
 \max_{\mathcal T_\tau}\mathcal U_\tau=\tau^2/3.
\label{eq:torsion-maximum}
\end{equation}
In either chart, put $w=u-4\mathcal U_{\tau}$ on
$\mathcal T_{\tau}$.  The two Laplacian estimates give
\[
 \Delta_x w=\Delta_xu-4\Delta_x\mathcal U_{\tau}
           \geq-4+4=0
\]
in the sense of distributions.  Equivalently, $w$ is the sum of
the subharmonic function $n^{-1}\log|P_\psi|^2$ and a smooth function
with nonnegative Laplacian.  Thus $w$ is subharmonic and upper
semicontinuous, with value $-\infty$ at polynomial zeros.
For each boundary point
$\xi\in\partial\mathcal T_{\tau}$, the chart image $\chi_\pm(\xi)$
belongs to $\widetilde K_t$, so $W_{\psi,n}(\xi)\leq M$.
Continuity of $W_{\psi,n}$ and the zero boundary values of
$\mathcal U_{\tau}$ imply
\[
 \limsup_{\substack{x\to\xi\\x\in\mathcal T_{\tau}}}w(x)\leq0.
\]
This remains valid when $W_{\psi,n}(\xi)=0$, because then
$u(x)\to-\infty$.

The maximum principle for a subharmonic function on a bounded domain
states that these nonpositive boundary limits imply $w\leq0$
throughout the domain.  It applies also at the corners of the triangle;
no boundary differentiability is needed for this form of the principle.
Thus \cref{eq:torsion-maximum} gives
\[
 u(x)\leq4\mathcal U_{\tau}(x)\leq\frac43\tau^2,
 \qquad x\in\mathcal T_{\tau}.
\]
Exponentiating \cref{eq:logarithmic-growth} yields
$W_{\psi,n}(x)\leq M\exp(4n\tau^2/3)$ in both chart triangles.
Their images contain $\CP^1\setminus\widetilde K_t$, while
$W_{\psi,n}\leq M$ on $\widetilde K_t$ by definition.  This proves
\cref{eq:triangle-growth}.  The assertion for $\psi=0$ is immediate.
\end{proof}

\begin{proposition}[Coherent-state lower bound]
\label[proposition]{prop:toeplitz}
For every fixed $0<t\leq10^{-3}$, there is a constant
$0<C_t<\infty$ such that
\begin{equation}
 \lambda_{\min}(T_n(\1_{\widetilde K_t}))
 \geq C_t^{-1}\e^{-1.35nt^2}
\label{eq:toeplitz-lower-bound}
\end{equation}
for all $n\geq1$.
\end{proposition}

This estimate fits the general lowest-eigenvalue asymptotics of Toeplitz
operators described in \cite[Theorem~1.1]{Finski}.  Here the triangle
barrier gives an explicit upper bound for the relevant weighted extremal
growth, and \cref{thm:BM} converts it into the stated lower bound.

\begin{proof}
Fix $t$ and the associated complement $\widetilde K_t$ of the two
open triangles.  Its Lipschitz boundary permits the Bernstein--Markov
estimate of \cref{thm:BM}.  Choose
\[
 \zeta=\left(1.35-\frac43(1.001)^2\right)t^2
       \simeq0.0139987t^2>0,
 \qquad C_t:=C_{\widetilde K_t,\zeta}.
\]
This constant is independent of $n$ and $\psi$.

For $\psi\in\mathcal H_n$, the resolution of the identity in \cref{eq:coherent-resolution}
and the fact that $\mu_{FS}$ is a probability measure give
\[
 \frac{\norm\psi^2}{N_n}
 =\int_{\CP^1}W_{\psi,n}\,\dd\mu_{FS}
 \leq\sup_{\CP^1}W_{\psi,n}.
\]
The triangle barrier in \cref{eq:triangle-growth}, followed by the Bernstein--Markov
inequality in \cref{eq:bernstein-markov}, yields
\[
\begin{aligned}
 \sup_{\CP^1}W_{\psi,n}
 &\leq\e^{\frac43n(1.001t)^2}\sup_{\widetilde K_t}W_{\psi,n}\\
 &\leq C_t\e^{n[\frac43(1.001t)^2+\zeta]}
             \int_{\widetilde K_t}W_{\psi,n}\,\dd\mu_{FS}\\
 &=C_t\e^{1.35nt^2}
             \int_{\widetilde K_t}W_{\psi,n}\,\dd\mu_{FS}.
\end{aligned}
\]
Combining these estimates and multiplying by $N_n$, the definition in
\cref{eq:toeplitz-definition} gives
\[
 \langle\psi,T_n(\1_{\widetilde K_t})\psi\rangle
 =N_n\int_{\widetilde K_t}W_{\psi,n}\,\dd\mu_{FS}
 \geq C_t^{-1}\e^{-1.35nt^2}\norm\psi^2.
\]
Taking the infimum over unit vectors proves \cref{eq:toeplitz-lower-bound}.
\end{proof}

\subsection{Binomial leakage and the frame bound}
\label{subsec:leakage}

Assign each point of $\widetilde K_t$ to a basis with the smallest
spin-one probability, resolving ties by the smallest index.  Explicitly, set
\begin{equation}
\begin{aligned}
 a_*([z])&:=\min\{a\in\{0,1,2\}:u_a(z)=\min_{b\in\{0,1,2\}}u_b(z)\},\\
 E_a&:=\{[z]\in\widetilde K_t:a_*([z])=a\},\qquad a=0,1,2.
\end{aligned}
\label{eq:leakage-partition}
\end{equation}
Continuity of the $u_a$ makes these sets measurable, and the tie rule gives
the disjoint partition $\widetilde K_t=E_0\sqcup E_1\sqcup E_2$.
Since $\widetilde K_t\subset K_t$, one has $u_a\leq1/2-t$ on $E_a$;
see \cref{fig:triangle}(b).

\begin{lemma}[Coherent mass outside a frequency cap]
\label[lemma]{lem:leakage}
If $0<\eta<t\leq10^{-3}$, then uniformly for $[z]\in E_a$,
\begin{equation}
 \norm{(I_{\mathcal H_n}-P_{\eta,n}^{(a)})e_z^{(n)}}^2
 \leq\exp[-2n(t-\eta)^2].
\label{eq:binomial-leakage}
\end{equation}
\end{lemma}

\begin{proof}
The squared Dicke coefficients of a coherent vector in the $a$-th basis
are binomial probabilities.  The left side is therefore its \emph{leakage}
outside the frequency cap, namely
\begin{equation}
 \Pr\{\operatorname{Bin}(n,u_a(z))/n>1/2-\eta\}.
\label{eq:leakage-event}
\end{equation}
For a Bernoulli random variable $X\in\{0,1\}$ of mean $q\in[0,1]$,
define $\Lambda:\mathbb R\to\mathbb R$ by
$\Lambda(\lambda)=\log\mathbb E\e^{\lambda(X-q)}$.
We have $\Lambda(0)=\Lambda'(0)=0$, and $\Lambda''(\lambda)\leq1/4$,
since it is the variance of a Bernoulli variable under exponential tilting.
Thus $\Lambda(\lambda)\leq\lambda^2/8$.  Exponential Markov for a sum
of $n$ independent copies, with $\lambda=4(t-\eta)$, yields
\[
 \Pr\{\operatorname{Bin}(n,q)/n-q>t-\eta\}
 \leq\e^{-2n(t-\eta)^2}.
\]
Use $q=u_a(z)\leq1/2-t$.  The degenerate cases $q=0,1$ also satisfy
the moment bound.
\end{proof}

\begin{lemma}[Coherent-to-type comparison]
\label[lemma]{lem:comparison}
For $0<\eta<t\leq10^{-3}$, write $r_n=\e^{-n(t-\eta)^2}>0$.
Then one has the inequality in
$\operatorname{End}(\mathcal H_n)$
\begin{equation}
 T_n(\1_{\widetilde K_t})
 \leq2N_nL_{\eta,n}+2N_nr_n^2I_{\mathcal H_n}.
\label{eq:toeplitz-frame-comparison}
\end{equation}
\end{lemma}

\begin{proof}
For $[z]\in E_a$, split the coherent vector with $P_{\eta,n}^{(a)}$
and its complement.  Cauchy--Schwarz and \cref{lem:leakage} give
\[
 |\langle e_z^{(n)},\psi\rangle|
 \leq\norm{P_{\eta,n}^{(a)}\psi}+r_n\norm\psi.
\]
Square using $(a+b)^2\leq2a^2+2b^2$, integrate over the partition, and
use $\sum_a\mu_{FS}(E_a)=\mu_{FS}(\widetilde K_t)\leq1$ to obtain
\[
\begin{aligned}
 \langle\psi,T_n(\1_{\widetilde K_t})\psi\rangle
 &\leq2N_n\sum_a\mu_{FS}(E_a)
               \norm{P_{\eta,n}^{(a)}\psi}^2
     +2N_nr_n^2\norm\psi^2\sum_a\mu_{FS}(E_a)\\
 &\leq2N_n\langle\psi,L_{\eta,n}\psi\rangle
     +2N_nr_n^2\norm\psi^2.
\end{aligned}
\]
This proves the operator inequality.
\end{proof}

\begin{proof}[Proof of \cref{prop:frame}]
Fix $0<\eta\leq10^{-4}$ and put $t=10\eta\leq10^{-3}$.
The coherent-state lower bound and the leakage estimate have exponent
coefficients
\begin{equation}
 1.35t^2=135\eta^2,
 \qquad 2(t-\eta)^2=162\eta^2.
\label{eq:exponent-margin}
\end{equation}
The operator comparison in \cref{lem:comparison} reads
\[
 L_{\eta,n}\geq\frac{1}{2N_n}T_n(\1_{\widetilde K_t})
                  -\e^{-2n(t-\eta)^2}I_{\mathcal H_n}.
\]
By \cref{prop:toeplitz}, with $C_\eta:=C_{10\eta}$,
\begin{equation}
 \lambda_{\min}(L_{\eta,n})
 \geq\frac{C_\eta^{-1}}{2(n+1)}\e^{-135n\eta^2}
       -\e^{-162n\eta^2}.
\label{eq:frame-lower-comparison}
\end{equation}
The strict inequalities $135<140<162$ allow the constant and $n+1$
to be absorbed, while the subtracted term decays strictly faster.
More precisely, divide the right side by $\e^{-140n\eta^2}$: the first
term tends to infinity and the second to zero.  The result is at least
one for all sufficiently large $n$, proving \cref{eq:frame-bound}.
\end{proof}

The numerical constants are not optimized.  In particular, the same
comparison gives every frame exponent coefficient strictly greater than
$135$; the value $140$ leaves a convenient margin.  Optimizing the
ratio $t/\eta$ could improve this further, but is unnecessary for the
counterexample.

\section{Completion of the counterexample}
\label{sec:contradiction}

We now choose a fixed frequency tolerance, then let the number of copies
tend to infinity.

\begin{proof}[Proof of \cref{thm:main}]
Recall the entropy deficit $\Delta(\eta)=h-f(1/2-\eta)$ from
\cref{eq:deficit}.  Fix
\begin{equation}
 \eta_* =10^{-4},\qquad
 \delta_0=\frac{\Delta(\eta_*)}{2},\qquad
 r=\frac{\Delta(\eta_*)}{4}-140\eta_*^2.
\label{eq:fixed-parameters}
\end{equation}
Since $\Delta(\eta_*)\geq\eta_*\log3$,
\begin{equation}
 r\geq\frac{\eta_*\log3}{4}-140\eta_*^2>0.
\label{eq:positive-rate}
\end{equation}
Since $\e^{-n(h-\delta)}\leq\e^{-n(h-\delta_0)}$ for
$0<\delta\leq\delta_0$, it suffices to prove the endpoint
$\delta=\delta_0$.
Let $\sigma_n$ satisfy \cref{eq:main-purity} with $\delta=\delta_0$.
Let $\tau_n=J_n^*\sigma_nJ_n$ be the compression from
\cref{lem:singlet-compression}.  By \cref{lem:cap},
\begin{equation}
 \Tr(\tau_nL_{\eta_*,n})
 \leq3\sqrt{n+1}\,\e^{-n\Delta(\eta_*)/4}.
\label{eq:fixed-cap-bound}
\end{equation}
For all sufficiently large $n$, \cref{prop:frame} gives
\begin{equation}
 L_{\eta_*,n}\geq\e^{-140n\eta_*^2}I_{\mathcal H_n}.
\label{eq:fixed-frame-bound}
\end{equation}
Combining the two estimates against the positive operator $\tau_n$,
\begin{equation}
 \e^{-140n\eta_*^2}\Tr\tau_n
 \leq\Tr(\tau_nL_{\eta_*,n})
 \leq3\sqrt{n+1}\,\e^{-n\Delta(\eta_*)/4},
\label{eq:trace-squeeze}
\end{equation}
whence
\begin{equation}
 \Tr\tau_n\leq3\sqrt{n+1}\,\e^{-rn}.
\label{eq:compressed-trace-decay}
\end{equation}
The overlap is preserved by the reduction, so \cref{eq:compressed-overlap} yields
\begin{equation}
 \langle\Omega^{\otimes n},\sigma_n\Omega^{\otimes n}\rangle
 \leq3\sqrt{n+1}\,\e^{-rn}\longrightarrow0.
\label{eq:overlap-decay}
\end{equation}
The measurement of the target projector and its complement gives
\begin{equation}
 \norm{\sigma_n-\proj{\Omega}^{\otimes n}}_1
 \geq2\left(1-\langle\Omega^{\otimes n},
                    \sigma_n\Omega^{\otimes n}\rangle\right)\longrightarrow2.
\label{eq:trace-separation}
\end{equation}
Since the trace norm between states is at most $2$, the limit is $2$.
All sufficiently-large-$n$ thresholds are independent of the candidate
state, so \cref{eq:overlap-decay} is a uniform bound over all candidates satisfying
\cref{eq:main-purity}.

For normalized marginals, operator-norm bounds imply purity bounds by
\cref{eq:purity-operator}.  Thus the operator-norm version fails as well.  Finally,
\cref{eq:phase-distance} shows that a pure-vector smoothing with tolerance
$\eps<\sqrt2$ would have
$|\langle\Omega^{\otimes n},\psi_n\rangle|\geq1-\eps^2/2>0$,
contradicting the overlap bound.
\end{proof}

\begin{remark}[Explicit parameters]
\label[remark]{rem:sizes}
The choices in \cref{eq:fixed-parameters} give
\[
 \delta_0\simeq5.494061443\cdot10^{-5},\qquad
 r\simeq2.607030722\cdot10^{-5}.
\]
\end{remark}

\begin{corollary}[Failure in all larger local dimensions and numbers of parties]
\label[corollary]{cor:larger}
For every $d\geq2$ and $m\geq4$, there is a pure state on
$(\mathbb C^d)^{\otimes m}$ that violates both simultaneous typicality
conjectures and their pure-vector variants.
\end{corollary}

\begin{proof}
First let $m=4$, fix an isometry $Q:\mathbb C^2\to\mathbb C^d$, and
write $\Omega_d=Q^{\otimes4}\Omega\in(\mathbb C^d)^{\otimes4}$.
Its three relevant pair entropies remain $h$.
Suppose that, for some $\eps<2$, states
$\sigma_n\in\operatorname{End}(((\mathbb C^d)^{\otimes4})^{\otimes n})$
stay within trace distance $\eps$ of $\proj{\Omega_d}^{\otimes n}$ and
have the three pair purities bounded by
$\e^{-n(h-\delta_0/2)}$.  Pull back to the qubit space:
\[
 \tau_n=(Q^*)^{\otimes4n}\sigma_nQ^{\otimes4n}
 \in\operatorname{End}(\cK^{\otimes n}),\qquad t_n=\Tr\tau_n.
\]
The target overlap is preserved, so \cref{eq:trace-overlap} gives
\[
 1\geq t_n\geq
 \langle\Omega^{\otimes n},\tau_n\Omega^{\otimes n}\rangle
 \geq1-\frac{\eps}{2}.
\]
For each $XY\in\{AB,AC,BC\}$, testing against positive operators gives
the inequality on the qubit pair space $(\mathbb C^2)^{\otimes2n}$
\[
 0\leq(\tau_n)_{(XY)^n}
 \leq(Q^*)^{\otimes2n}(\sigma_n)_{(XY)^n}Q^{\otimes2n}.
\]
Hilbert--Schmidt monotonicity from \cref{eq:positive-norm-monotonicity} and contractivity under
compression therefore imply, for the state
$\widehat\tau_n=\tau_n/t_n\in\operatorname{End}(\cK^{\otimes n})$,
\[
 \Tr\!\left[(\widehat\tau_n)_{(XY)^n}^2\right]
 \leq(1-\eps/2)^{-2}\e^{-n(h-\delta_0/2)}
 \leq\e^{-n(h-\delta_0)}
\]
for sufficiently large $n$, absorbing the fixed normalization factor
with the slack $\delta_0/2$.  Yet its target overlap is at least
$1-\eps/2>0$, contradicting \cref{thm:main}.

For $m>4$, tensor $\Omega_d$ with a fixed pure product state
on $(\mathbb C^d)^{\otimes(m-4)}$.  A hypothetical smoothing may first be traced over
those parties: trace distance contracts, and the three obstructing pair
marginals are unchanged.  The resulting four-party state contradicts the
case just proved.  Finally, \cref{eq:phase-distance} and the overlap
argument give the pure-vector assertion for every $\eps<\sqrt2$.
\end{proof}

\section{Robustness}
\label{sec:discussion}

The construction establishes failure for pure states on four parties and,
by \cref{cor:larger}, for every larger number and every local dimension at
least two.  \Cref{cor:mixed-tripartite} also turns the same construction
into a mixed three-party counterexample.  On the other hand, the pure
tripartite case is positive by \cref{rem:pure-tripartite}.  Thus four parties
are exactly the threshold for pure inputs, while three parties already suffice
for mixed inputs.  At the one-shot level, \cref{cor:one-shot} rules out both
the original simultaneous min-entropy-smoothing conjecture and its conditional
variant.

The symmetry of the chosen state makes the frame estimate accessible,
but the resulting obstruction persists when that symmetry is broken.
The following proposition makes this stability precise.

\Needspace{8\baselineskip}
\begin{proposition}[An open neighborhood of counterexamples]
\label[proposition]{prop:robustness}
Put $\rho_\Omega=\proj{\Omega}\in\operatorname{End}(\cK)$, where
$\cK=V^{\otimes4}=(\mathbb C^2)^{\otimes4}$ is the physical four-qubit
space from \cref{sec:candidate}.  There exist
$\eps_*>0$ and $\delta_*>0$ such that, for every density operator
$\rho\in\operatorname{End}(\cK)$ with $\norm{\rho-\rho_\Omega}_1<\eps_*$,
every sequence of states $\sigma_n\in\operatorname{End}(\cK^{\otimes n})$
satisfying, for all sufficiently large $n$,
\[
 \Tr\!\left[(\sigma_n)_{T^n}^2\right]
 \leq\e^{-n(H(\rho_T)-\delta_*)},
 \qquad T\in\{AB,AC,BC\},
\]
obeys
\[
 \norm{\sigma_n-\rho^{\otimes n}}_1\longrightarrow2.
\]
The same conclusion holds with operator norms in place of purities.
In particular, both simultaneous typicality conjectures fail throughout
this neighborhood, with no purity or symmetry assumption on $\rho$.
\end{proposition}

\begin{proof}
\emph{A uniform estimate and a neighborhood of the target.}
Use the fixed $\eta_*,\delta_0,r$ of \cref{eq:fixed-parameters}; in particular
$0<\delta_0<h$ and $r>0$.
The quantitative conclusion of \cref{thm:main} is uniform;
we weaken its polynomial prefactor to $3(m+1)$ for convenience:
there is $m_0$ such that every state
$\zeta_m\in\operatorname{End}(\cK^{\otimes m})$ satisfying
\[
 \Tr\!\left[(\zeta_m)_{T^m}^2\right]
 \leq\e^{-m(h-\delta_0)}\quad(T\in\{AB,AC,BC\})
\]
has
\[
 \Tr(\rho_\Omega^{\otimes m}\zeta_m)
 \leq3(m+1)\e^{-rm}\qquad(m\geq m_0).
\]
Here $m_0$ is the uniform threshold supplied by \cref{thm:main}.

Recall the binary entropy $b:[0,1]\to\mathbb R$ from \cref{eq:binary-entropy}.
Choose $0<\alpha<1/2$ so small that
\[
 \alpha(h+\log4)\leq\frac{\delta_0}{2},
 \qquad b(\alpha)<r(1-\alpha),
\]
and set $\delta_*=\delta_0/4$.
Continuity of the overlap and of finite-dimensional marginal entropies
gives $\eps_*>0$ such that $\norm{\rho-\rho_\Omega}_1<\eps_*$ implies
\[
 p:=1-\Tr(\rho_\Omega\rho)<\alpha,
 \qquad H(\rho_T)>h-\frac{\delta_0}{4}
 \quad(T\in\{AB,AC,BC\}).
\]
Indeed, these strict inequalities hold at $\rho=\rho_\Omega$, and define
an open set relative to the density operators.  For such a $\rho$,
the assumed purity bounds imply
\[
 \Tr\!\left[(\sigma_n)_{T^n}^2\right]
 \leq\e^{-n(h-\delta_0/2)}.
\]

\emph{Few defects have small weight in every admissible state.}
Set $k=\lfloor\alpha n\rfloor$ and $m=n-k$, and label the physical
copies $\cK_i\cong\cK$, $i=1,\ldots,n$.
For any set $J\subseteq\{1,\ldots,n\}$ of size $k$, discard the
copies in $J$ and call the resulting state
$\zeta_{J^c}\in\operatorname{End}(\bigotimes_{i\notin J}\cK_i)
\cong\operatorname{End}(\cK^{\otimes m})$, with the remaining copies in
their original order.  For finite-dimensional Hilbert spaces $E,F$ and
$X\in\operatorname{End}(E\otimes F)$, Hilbert--Schmidt duality gives
\[
 \norm{\Tr_F X}_2
 =\sup_{Y\in\operatorname{End}(E),\,\norm Y_2=1}|\Tr[(Y^*\otimes I_F)X]|
 \leq\sqrt{\dim F}\,\norm X_2.
\]
Discarding $k$ copies of a two-qubit marginal therefore costs at
most $4^k$ in purity, so
\[
 \Tr\!\left[(\zeta_{J^c})_{T^m}^2\right]
 \leq4^k\e^{-n(h-\delta_0/2)}
 \leq\e^{-m(h-\delta_0)}.
\]
The last inequality is equivalent to
$k(\log4-h+\delta_0)\leq n\delta_0/2$.  This follows from
$k\leq\alpha n$ and $\alpha(h+\log4)\leq\delta_0/2$, since
$0<\delta_0<h$.
Thus, for all sufficiently large $n$, uniformly in $J$,
\[
 \Tr\!\left[\sigma_n
   \left(\bigotimes_{i\notin J}\rho_{\Omega,i}\right)
   \left(\bigotimes_{i\in J}I_i\right)\right]
 \leq3(m+1)\e^{-rm}.
\]
Here $\rho_{\Omega,i},I_i\in\operatorname{End}(\cK_i)$ are the corresponding
copies of $\rho_\Omega,I_{\cK}$.

Let $Q_{n,k}\in\operatorname{End}(\cK^{\otimes n})$ be the orthogonal
projection onto at most $k$ defects, where a
\emph{defect} means the outcome $I_i-\rho_{\Omega,i}$ on copy $i$:
\[
 Q_{n,k}=\sum_{\substack{S\subseteq\{1,\ldots,n\}\\|S|\leq k}}
   \left(\bigotimes_{i\in S}(I_i-\rho_{\Omega,i})\right)
   \left(\bigotimes_{i\notin S}\rho_{\Omega,i}\right).
\]
These single-copy projections commute.  Each defect set $S$ with
$|S|\leq k$ is contained in some $J$ of size $k$, which gives
the operator inequality
\[
 Q_{n,k}\leq\sum_{|J|=k}
   \left(\bigotimes_{i\notin J}\rho_{\Omega,i}\right)
   \left(\bigotimes_{i\in J}I_i\right).
\]
Consequently,
\[
 \Tr(\sigma_nQ_{n,k})
 \leq\binom nk3(m+1)\e^{-rm}
 \leq3(n+1)\e^{-n[r(1-\alpha)-b(\alpha)]}
 \longrightarrow0.
\]
We used $\binom nk\leq\e^{nb(k/n)}\leq\e^{nb(\alpha)}$, since
$k/n\leq\alpha<1/2$.  The first inequality follows from the
$k$-th term of the binomial probability sum with parameter $k/n$.

\emph{The nearby tensor power has few defects.}
In $\rho^{\otimes n}$, the same measurements have independent defect
probability $p<\alpha$, even when $\rho$ does not commute with
$\rho_\Omega$.  The law of large numbers therefore gives
\[
 \Tr(\rho^{\otimes n}Q_{n,k})
 =\Pr\{\operatorname{Bin}(n,p)\leq k\}\longrightarrow1.
\]
For states $a,b\in\operatorname{End}(E)$ and an orthogonal projection
$Q\in\operatorname{End}(E)$,
$\norm{a-b}_1\geq2|\Tr[Q(a-b)]|$, since the positive and negative
parts of the trace-zero operator $a-b$ each have trace
$\norm{a-b}_1/2$.  Applying this with $Q=Q_{n,k}$ yields
\[
 \norm{\sigma_n-\rho^{\otimes n}}_1
 \geq2\bigl(\Tr(\rho^{\otimes n}Q_{n,k})
             -\Tr(\sigma_nQ_{n,k})\bigr)\longrightarrow2.
\]
The trace norm between states is at most $2$, proving the limit.
Finally, $\Tr a^2\leq\norm a_\infty$ for a density operator $a$,
so the operator-norm constraints imply the purity constraints.
\end{proof}

In particular, $(1-s)\rho_\Omega+sI_{\cK}/16$ is a full-rank counterexample for
all sufficiently small $s>0$.

The smallest nontrivial local dimension already supports the mechanism.
The two-dimensional singlet multiplicity admits three cyclically related
measurement bases.  The two conjugate eigendirections of the cyclic
matrix $\Gamma$ give the two entropy equality points.  The single complex coordinate makes the
triangular geometry explicit; a one-sided binary cap gives a uniform
binomial leakage estimate without a separate entropy-rate minimization.

\begin{remark}[The quadratic order of the frame cost]
\label[remark]{rem:optimal}
Evaluating the frame on its target vector gives
\begin{equation}
 \lambda_{\min}(L_{\eta,n})
 \leq\langle v^{\otimes n},L_{\eta,n}v^{\otimes n}\rangle
 =3\Pr\{\operatorname{Bin}(n,1/2)/n\leq1/2-\eta\}.
\label{eq:target-rayleigh}
\end{equation}
By Stirling's formula, for fixed $\eta$ the right side has exponential
decay rate $\log2-b(1/2-\eta)=2\eta^2+O(\eta^4)$, so the quadratic
order of the frame exponent cannot be replaced by $o(\eta^2)$.
\end{remark}

\section*{Acknowledgements}

The author acknowledges support from the Ministry of Education, Culture,
Sports, Science and Technology (MEXT) through SPReaD (Supporting
Pioneering Research through AI for 1,000 Discovery challenges).  The author
also thanks P.~Hayden, Y.~Ogata and N.~Ozawa for inspiring discussions
about quantum typicality, the literature and methodology.

An earlier, more complicated qutrit construction was developed with the
assistance of ChatGPT 5.6.  The approach to its frame bound used
pluripotential theory and Berezin--Toeplitz operators, drawing on Finski's
eigenvalue estimates \cite{Finski} and the Bernstein--Markov criterion of
Berman, Boucksom and Witt Nystr\"om \cite{BermanBoucksomWittNystrom}.
The qubit state studied here was obtained through subsequent human
refinement and simplification.  LLMs from OpenAI and Anthropic were used
to assist with editing, preparing illustrations and checking parts of this manuscript,
and also drew the author's attention to its relation to the
Higuchi--Sudbery state.
All mathematical content was reviewed, checked and substantially revised
by the author, who takes full responsibility for it.

\appendix
\crefalias{section}{appendix}
\section{Exact matrices of the counterexample}
\label{app:matrices}

We retain the physical tensor order $A,B,C,R$, with each party equal
to $\mathbb C^2$, and write $M_d(\mathbb C)=\operatorname{End}(\mathbb C^d)$.
Every computational basis is ordered lexicographically: for a pair
$XY$, it is
$\mathcal B_{XY}=(\ket{00},\ket{01},\ket{10},\ket{11})$, and for
$ABC$ it is
\[
 \mathcal B_{ABC}=(\ket{000},\ket{001},\ket{010},\ket{011},
                  \ket{100},\ket{101},\ket{110},\ket{111}).
\]
For each subsystem $T$, put
$\rho_T=\Tr_{T^c}\proj{\Omega}\in\operatorname{End}(T)$.
Throughout this appendix, $\omega$ is the scalar
\[
 \omega=\e^{2\pi i/3}=-\frac12+\frac{i\sqrt3}{2},
 \qquad \bar\omega=\omega^2,\qquad 1+\omega+\omega^2=0.
\]
With this notation, all entries below are exact.

\Needspace{14\baselineskip}
\subsection{The four-party vector as a matrix}

Group the first two and last two indices of $\Omega$.  Its \emph{coefficient
matrix} $M_{AB|CR}\in M_4(\mathbb C)$, with rows indexed by $ab$ and
columns by $cr$, is defined by
$\Omega=\sum_{a,b,c,r=0}^1(M_{AB|CR})_{ab,cr}\ket{abcr}$.
\Cref{eq:target-computational} gives
\begin{equation}
 M_{AB|CR}=\frac{-i}{\sqrt6}
 \begin{pmatrix}
  0&0&0&1\\
  0&\omega&\omega^2&0\\
  0&\omega^2&\omega&0\\
  1&0&0&0
 \end{pmatrix}.
\label{eq:app-coefficient-matrix}
\end{equation}
As a tensor, the ket $\Omega$ defines a linear map
$(C\otimes R)^*\to A\otimes B$, whose matrix in the dual computational
basis of the domain and $\mathcal B_{AB}$ is $M_{AB|CR}$.
Here the star on a space denotes its \emph{complex-linear dual}.

\subsection{The three-party and pair marginals}

The two vectors in $A\otimes B\otimes C$ given by
\begin{equation}
 \xi_0=\frac{\ket{110}+\omega\ket{101}+\omega^2\ket{011}}{\sqrt3},
 \qquad
 \xi_1=\frac{\ket{001}+\omega\ket{010}+\omega^2\ket{100}}{\sqrt3}
\label{eq:app-schmidt-vectors}
\end{equation}
are orthonormal.  Regrouping \cref{eq:target-computational} by the last qubit gives the
\emph{Schmidt decomposition}, up to an overall phase,
\[
 \Omega=\frac{-i}{\sqrt2}
   (\xi_0\otimes\ket{0}_R+\xi_1\otimes\ket{1}_R),\qquad
 \rho_{ABC}=\frac12(\proj{\xi_0}+\proj{\xi_1}).
\]
Consequently, in the basis $\mathcal B_{ABC}$,
\begin{equation}
 [\rho_{ABC}]_{\mathcal B_{ABC}}=\frac16
 \begin{pmatrix}
  0&0&0&0&0&0&0&0\\
  0&1&\omega^2&0&\omega&0&0&0\\
  0&\omega&1&0&\omega^2&0&0&0\\
  0&0&0&1&0&\omega&\omega^2&0\\
  0&\omega^2&\omega&0&1&0&0&0\\
  0&0&0&\omega^2&0&1&\omega&0\\
  0&0&0&\omega&0&\omega^2&1&0\\
  0&0&0&0&0&0&0&0
 \end{pmatrix}\in M_8(\mathbb C).
\label{eq:app-three-party-matrix}
\end{equation}
The two rank-one terms have disjoint computational supports, of Hamming
weights two and one, respectively.  Thus $\rho_{ABC}$ has rank two,
its nonzero eigenvalues are $1/2,1/2$, and
\begin{equation}
 \rho_{ABC}^2=\tfrac12\rho_{ABC},\qquad
 \Tr(\rho_{ABC}^2)=\tfrac12,\qquad H(\rho_{ABC})=\log2.
\label{eq:app-three-party-invariants}
\end{equation}

Taking the partial traces of \cref{eq:app-three-party-matrix} gives the
same matrix for all three pairs in their respective computational bases:
\begin{equation}
 [\rho_{XY}]_{\mathcal B_{XY}}=\frac16
 \begin{pmatrix}
  1&0&0&0\\
  0&2&-1&0\\
  0&-1&2&0\\
  0&0&0&1
 \end{pmatrix}\in M_4(\mathbb C),\qquad XY\in\{AB,AC,BC\}.
\label{eq:app-pair-matrix}
\end{equation}
The singlet $(\ket{01}-\ket{10})/\sqrt2$ is an eigenvector with eigenvalue
$1/2$, and its three-dimensional orthogonal complement is an eigenspace
with eigenvalue $1/6$.
Hence
\[
 \rank\rho_{XY}=4,\qquad \Tr(\rho_{XY}^2)=\frac13,\qquad
 H(\rho_{XY})=\log2+\frac12\log3=h.
\]
In particular, the singular values of $M_{AB|CR}$ are
$1/\sqrt2,1/\sqrt6,1/\sqrt6,1/\sqrt6$.
Every single-qubit marginal satisfies $\rho_X=I_X/2$,
$X\in\{A,B,C,R\}$, as follows by tracing either qubit in
\cref{eq:app-pair-matrix} and from the Schmidt decomposition for $R$.

\printbibliography
\end{document}